\documentclass[11pt,a4paper]{article}

\usepackage[utf8]{inputenc}
\usepackage[T1]{fontenc}
\usepackage{amsmath}
\usepackage{amsfonts}
\usepackage{amssymb}
\usepackage{amsthm}
\usepackage{lmodern}
\usepackage[hypertexnames=false]{hyperref}
\usepackage{tikz}
\usepackage[dvipsnames]{xcolor}
\usepackage{mathrsfs}
\usepackage{cite}
\usepackage[caption=false]{subfig}
\usepackage{tabularx}
\usepackage{xurl}
\usepackage[hmargin=3.75cm,vmargin=3.75cm]{geometry}
\usepackage{enumitem}
\usepackage{doi}
\usepackage{authblk}
\usepackage{orcidlink}

\hypersetup{colorlinks,allcolors=blue}

\setlist[enumerate,itemize]{noitemsep}

\newtheorem{theorem}{Theorem}
\newtheorem{lemma}{Lemma}
\newtheorem{definition}{Definition}
\newtheorem{proposition}{Proposition}
\newtheorem{fact}{Fact}
\theoremstyle{remark}
\newtheorem{example}{Example}
\newtheorem{remark}{Remark}

\usetikzlibrary{backgrounds,arrows.meta}

\colorlet{myred}{purple!30}
\colorlet{myblue}{CornflowerBlue!50}
\colorlet{mygreen}{YellowGreen!60}
\colorlet{mybgcolor}{yellow!20}
\colorlet{mynodecolor}{CadetBlue!15}

\newcommand{\Ag}[0]{\mathcal{A}}
\newcommand{\Var}[1]{\mathsf{Var}_{#1}}
\newcommand{\Atoms}[0]{\mathsf{Var}}
\newcommand{\LKB}[0]{\mathcal{L}_{B}^{K}}
\newcommand{\LB}[0]{\mathcal{L}_{B}}
\newcommand{\B}[1]{\mathord{B_{#1}}}
\newcommand{\K}[1]{\mathord{K_{#1}}}
\newcommand{\axiom}[2]{\text{(#1\textsubscript{$#2$})}}
\newcommand{\DB}[0]{\axiom{D}{B}}
\newcommand{\FourB}[0]{\axiom{4}{B}}
\newcommand{\FiveB}[0]{\axiom{5}{B}}
\newcommand{\KB}[0]{\axiom{K}{B}}
\newcommand{\KK}[0]{\axiom{K}{K}}
\newcommand{\TK}[0]{\axiom{T}{K}}
\newcommand{\FourK}[0]{\axiom{4}{K}}
\newcommand{\FiveK}[0]{\axiom{5}{K}}
\newcommand{\unaryaxiom}[1]{\text{(#1)}}
\newcommand{\KIB}[0]{\unaryaxiom{KB}}
\newcommand{\SPI}[0]{\unaryaxiom{SP}}
\newcommand{\SNI}[0]{\unaryaxiom{SN}}
\newcommand{\Loc}[0]{\unaryaxiom{Loc}}
\newcommand{\EDL}[0]{\mathsf{EDL}_{n}}
\newcommand{\LocKDFourFive}[0]{\mathsf{KD45}_{n}^{L}}
\newcommand{\LocKFourFive}[0]{\mathsf{K45}_{n}^{L}}
\newcommand{\Cond}[1]{\text{\upshape(P#1)}}
\newcommand{\KFourFive}[0]{\mathsf{K45}}
\newcommand{\KDFourFive}[0]{\mathsf{KD45}}
\newcommand{\SFive}[0]{\mathsf{S5}}
\newcommand{\DoxRels}[0]{R}
\newcommand{\EpRels}[0]{E}
\newcommand{\Brel}[1]{R_{#1}}
\newcommand{\Krel}[1]{E_{#1}}
\newcommand{\eqrel}[1]{#1^{\equiv}}
\newcommand{\sym}[1]{#1^{\leftrightarrow}}
\newcommand{\KTEclass}[0]{\mathcal{K}_{n}^{\mathit{te}}}
\newcommand{\KSTEclass}[0]{\mathcal{K}_{n}^{\mathit{ste}}}
\newcommand{\HSUclass}[0]{\mathcal{H}_{n}^{\mathit{su}}}
\newcommand{\HSUTclass}[0]{\mathcal{H}_{n}^{\mathit{sut}}}
\newcommand{\tail}[1]{T(#1)}
\newcommand{\head}[1]{H(#1)}
\newcommand{\undir}[1]{U(#1)}
\newcommand{\edge}[0]{E}
\newcommand{\Vertices}[0]{\mathscr{V}}
\newcommand{\Edges}[0]{\mathscr{E}}
\newcommand{\rk}[1]{\mathit{rk}(#1)}
\newcommand{\colmap}[0]{\chi}
\newcommand{\GBrel}[1]{\vartriangleright_{#1}}
\newcommand{\GKrel}[1]{\approx_{#1}}
\newcommand{\andsymbol}[0]{\mathrel{\&}}
\newcommand{\MCS}[0]{\mathsf{MCS}}
\newcommand{\modform}[2]{#1/#2}
\newcommand{\Vform}[2]{V_{#1}(#2)}
\newcommand{\Bform}[2]{B_{#1}(#2)}
\newcommand{\locstate}[2]{{#1}_{#2}}
\newcommand{\Tset}[1]{T_{#1}}
\newcommand{\Hset}[1]{H_{#1}}
\newcommand{\Mc}[0]{M^{c}}
\newcommand{\Gc}[0]{G^{c}}
\newcommand{\COLc}[0]{\chi^{c}}
\newcommand{\VALc}[0]{\ell^{c}}
\newcommand{\CS}{\mathscr{S}}
\newcommand{\auxrel}[0]{\mathcal{R}}
\newcommand{\hypmod}[1]{\mathfrak{H}(#1)}
\newcommand{\kripmod}[1]{\mathfrak{K}(#1)}

\newcommand{\semK}{\models_{\text{\upshape\scshape k}}}
\newcommand{\semH}{\models_{\text{\upshape\scshape h}}}
\newcommand{\prov}{\vdash}
\newcommand{\notprov}{\nvdash}

\begin{document}

\title{Hypergraph Semantics for Doxastic Logics}

\author[1]{Hans van Ditmarsch\,\orcidlink{0000-0003-4526-8687}}
\author[2]{Djanira Gomes\,\orcidlink{0009-0002-1248-4164}}
\author[2]{David Lehnherr\,\orcidlink{0000-0002-4956-4064}}
\author[2]{Valentin~M{\"u}ller\,\orcidlink{0009-0008-7251-0831}}
\author[2]{Thomas Studer\,\orcidlink{0000-0002-0949-3302}}

\affil[1]{CNRS, IRIT, University of Toulouse, France}
\affil[2]{Institute of Computer Science, University of Bern, Switzerland}

\date{}

\maketitle

\begin{abstract}
\noindent\looseness=-1 Simplicial models have become a crucial tool for studying distributed computing. These models, however, are only able to account for the knowledge, but not for the beliefs of agents. We present a new semantics for logics of belief. Our semantics is based on directed hypergraphs, a generalization of ordinary directed graphs in which edges are able to connect more than two vertices. Directed hypergraph models preserve the characteristic features of simplicial models for epistemic logic, while also being able to account for the beliefs of agents. We provide systems of both consistent belief and merely introspective belief. The completeness of our axiomatizations is established by the construction of canonical hypergraph models. We also present direct conversions between doxastic Kripke models and directed hypergraph models.\medskip

\noindent\textbf{Keywords:} Doxastic logic, epistemic logic, modal logic, hypergraphs, simplicial complexes, distributed systems, canonical models.
\end{abstract}

\section{Introduction}\label{sec:introduction}

Traditionally, epistemic modal logic is interpreted using \emph{relational semantics} on labeled graphs. Vertices correspond to possible worlds, and an $a$-labeled edge between two worlds $v$ and $w$ indicates that agent $a$ cannot distinguish $v$ from $w$. Although expressive, relational models do not intrinsically describe what exactly constitutes a possible world. 

A prominent example of a relational semantics whose possible worlds have a structure is the \emph{interpreted systems model} by Fagin et al.~\cite{fagin:etal:2003}. In an interpreted system, a possible world is a \emph{global state}, which is composed of the \emph{local states} of all agents. A local state, in turn, is a sequence of locally observed events. However, this additional expressiveness is only achieved by extending the model, rather than using the structural properties of labeled graphs.

\looseness=-1 In contrast, the newly emerging \emph{simplicial semantics} for modal logic, proposed by Goubault et al.~\cite{goubault:etal:2018}, uses the structure of a simplicial complex to emphasize the local attributes of possible worlds. This approach was inspired by the work of Herlihy et al.~\cite{herlihy:etal:2014}, who pioneered the use of simplicial complexes to reason about the solvability of distributed tasks. Similar to interpreted systems, simplicial complexes describe a global state as a composition of local states. Formally, a simplicial complex is defined to be a pair $C = (\Vertices, \CS)$, where $\Vertices$ is a set of vertices and $\CS \neq \emptyset$ is a set of non-empty finite subsets of $\Vertices$ that is closed under set inclusion. Elements of $\CS$ are called \emph{faces} (or \emph{simplices}), and faces that are maximal with respect to set inclusion are referred to as \emph{facets}. A chromatic simplicial complex is a simplicial complex equipped with a coloring function that assigns an agent to each vertex in such a way that distinct vertices belonging to the same face are assigned different agents. Intuitively, an $a$-colored vertex denotes a possible local state of agent $a$. The vertices contained in a face are considered mutually compatible. Facets are interpreted as global states, i.e., maximal sets of compatible local states. If an agent's local state belongs to two facets, the corresponding global states are considered indistinguishable for that agent.

\begin{figure}[t]
\centering%
\subfloat[]{\label{subfig:intro:a}%
\scalebox{.7}{%
\begin{tikzpicture}[minimum size=.7cm]
\node[circle,draw,fill=myred] (a1) at (0,2) {$a$};
\node[circle,draw,fill=myblue] (b1) at (0,0) {$b$};
\node[circle,draw,fill=mygreen] (c2) at (2,1) {$c$};
\node[circle,draw,fill=mygreen] (c1) at (-2,1) {$c$};

\node (E1) at (barycentric cs:a1=1,b1=1,c1=1) {$\edge_1$};
\node (E2) at (barycentric cs:a1=1,b1=1,c2=1) {$\edge_2$};

\begin{scope}[on background layer]
\filldraw[fill=mybgcolor] (a1.center)--(b1.center)--(c1.center)--cycle;
\filldraw[fill=mybgcolor] (a1.center)--(b1.center)--(c2.center)--cycle;
\end{scope} 
\end{tikzpicture}}}%
\hspace*{5em}%
\subfloat[]{\label{subfig:intro:b}%
\scalebox{.7}{%
\begin{tikzpicture}[minimum size=.7cm]
\node[circle,draw,fill=mygreen] (c1) at (2,1) {$c$};
\node[circle,draw,fill=myred] (a2) at (4,2) {$a$};
\node[circle,draw,fill=myblue] (b2) at (4,0) {$b$};
\node[circle,draw,fill=mygreen] (c2) at (6,1) {$c$};

\coordinate[label={[shift={(0,0)}]$\edge_1$}] (E1) at (barycentric cs:a2=1,b2=1,c1=2.5);
\coordinate[label={[shift={(0,0)}]$\edge_2$}] (E2) at (barycentric cs:a2=1,b2=1,c2=2.5);

\draw[->] (c1) -- (E1) to[out=0,in=180] (a2);
\draw[->] (E1) to[out=0,in=180] (b2);

\draw[->] (a2) to[out=0,in=180] (E2) -- (c2);
\draw (b2) to[out=0,in=180] (E2);
\end{tikzpicture}}}%
\caption{An undirected hypergraph (left) and a directed hypergraph (right).}\label{fig:intro}
\end{figure}
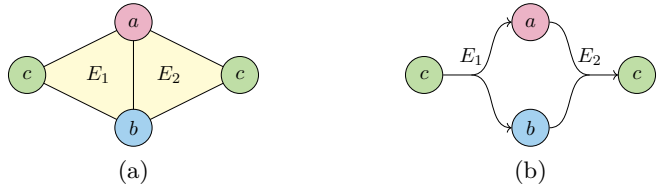

\emph{Undirected hypergraphs} are a generalization of simplicial complexes and were first proposed as models for knowledge by Goubault et al.~\cite{goubault:etal:2024}. An undirected hypergraph is a pair $G=(\Vertices,\Edges)$, where $\Vertices$ is a set of vertices and $\Edges\subseteq\mathcal{P}(\Vertices)$ is an arbitrary set of subsets of $\Vertices$. Elements of $\Edges$ are called \emph{undirected hyperedges}. As before, vertices correspond to local states. Hyperedges, on the other hand, represent global states. Local states are compatible with each other if they belong to the same hyperedge. Indistinguishability follows the same principle as in simplicial models. However, in contrast to simplicial complexes, the set of hyperedges of a hypergraph is not necessarily closed under set inclusion. Undirected hypergraphs thus allow for a more succinct representation of semantics based on local states. To illustrate this, consider the undirected hypergraph in Figure~\ref{subfig:intro:a}. The graph consists of four vertices and two hyperedges $\edge_1$ and $\edge_2$. Agents $a$ and $b$ cannot distinguish between $\edge_1$ and $\edge_2$, because their local states are the same in both. Agent $c$, on the other hand, is able to distinguish $\edge_1$ and $\edge_2$, since her local state in $\edge_1$ differs from her local state in $\edge_2$. Note that the simplicial complex representing the same situation would also have to contain all non-empty subsets of $\edge_1$ and $\edge_2$, which amounts to a total of eleven faces. 

A \emph{directed hypergraph} is a pair $G = (\Vertices, \Edges)$ consisting of a set of vertices $\Vertices$ and a set of \emph{directed hyperedges} $\Edges$. Unlike undirected hyperedges, a directed hyperedge is a pair $\edge = (X,Y)$ that contains a \emph{tail} $X \subseteq \Vertices$ and a \emph{head} $Y \subseteq \Vertices$. Throughout this paper, the tail and the head of a directed hyperedge are assumed to be disjoint. To the best of our knowledge, directed hypergraphs have not yet been studied from an epistemic point of view. In this paper, we will use directed hypergraphs as models for both consistent belief and merely introspective belief. The vertices of a directed hypergraph are taken to represent local states of agents and directed hyperedges are interpreted as global states. A global state $\edge_2$ is considered to be \emph{doxastically possible} by an agent in a global state $\edge_1$, if the agent's local state is contained in $\edge_1$ and belongs to the tail of $\edge_2$. As a consequence, the tail of a directed hyperdege $\edge$ may be conceived as the set of all local states whose associated agents consider $\edge$ possible when in $\edge$. Figure~\ref{subfig:intro:b} shows an example of a directed hypergraph with two hyperedges $\edge_1$ and $\edge_2$. Incoming arrows indicate that a vertex belongs to the head of a hyperedge, while outgoing arrows indicate that it belongs to the tail. For instance, when in $\edge_1$, agent $c$ considers $\edge_1$ possible, but the other agents do not. This is because only the local state of agent $c$ is part of the tail of $\edge_1$. Moreover, agent $a$ considers $\edge_2$ possible when in $\edge_1$, since her local state is included in the tail of $\edge_2$. However, agent $a$ does not consider $\edge_1$ possible when in $\edge_1$ or $\edge_2$, as her local state does not belong to the tail of $\edge_1$. This shows that possibility relations on directed hypergraphs need not be reflexive or symmetric. 

So far, we have not discussed reasoning about knowledge on simplicial complexes or hypergraphs. The procedure is the same for both: atoms may be assigned to either local states or global states. Randrianomentsoa et al.~\cite{randrianomentsoa:etal:2023} assign them to local states. Goubault et al.~\cite{goubault:kniazev:etal:2024,goubault:etal:2024} assign them to global states in the simplicial case, and to a mixture of both in the case of undirected hypergraphs. Regardless of which option is chosen, knowledge can be described in the usual way: an agent $a$ knows a formula $\varphi$ if and only if $\varphi$ is true in each global state that $a$ cannot distinguish from the current global state. Numerous works have analyzed different types of simplicial complexes and their relationship to various notions of (distributed) knowledge \cite{ledent:2019,ditmarsch:etal:2021,goubault:etal:2021,ditmarsch:etal:2022,goubault:etal:2022,goubault:etal:2023,randrianomentsoa:etal:2023,bilkova:etal:2024,cachin:etal:2025:a,cachin:etal:2025:b,randrianomentsoa:2025}.

Unlike knowledge, belief is generally considered to be non-factive. That is, an agent can believe a formula $\varphi$ although $\varphi$ might be false in the current state. In relational semantics, false beliefs can be addressed by simply omitting the reflexivity condition on the indistinguishability relations or by replacing it with a seriality condition. While relatively unproblematic for relational semantics, non-factivity poses a major challenge for simplicial semantics (cf.~\cite[Sect.~4.3]{castaneda:etal:2024}). In fact, in simplicial semantics, the agent's true local state must be contained in the face representing the actual global state, so the actual state is always considered possible. Hence, an indistinguishability relation on simplicial complexes that is solely based on connectivity is inherently reflexive. Cachin et al.~\cite{cachin:etal:2025:a} address this problem by studying simplicial complexes in which adjacent vertices may have the same color. This allows for a non-reflexive possibility relation. However, the philosophical interpretation of an agent simultaneously having multiple local states remains an open problem. This difficulty does not arise for models based on directed hypergraphs. Their interpretation is intuitive and allows for a possibility relation that is not necessarily reflexive or symmetric. Thus, they are compelling structures that are well suited for modeling various notions of belief. The present work also generalizes and extends the framework of directed simplicial complexes introduced by Lehnherr in his forthcoming PhD thesis \cite{lehnherr:forthcoming}.

\looseness=-1 The paper is structured as follows. In Section~\ref{sec:preliminaries}, we first recall some basic notions and introduce three axiomatic systems for multi-agent doxastic logic: a system of consistent belief, a system of merely introspective belief and a system of epistemic-doxastic logic. In Section~\ref{sec:hypergraph:semantics}, we use directed hypergraphs to define a new semantics for doxastic logics, thereby addressing the open question posed in \cite[Sect.~4.3]{castaneda:etal:2024}. In Section~\ref{sec:completeness}, we then show that our axiomatizations are sound and complete with respect to certain classes of directed hypergraph models. Our completeness proof is based on a canonical model construction. In Section~\ref{sec:conversions:kripke:hypergraph}, we provide direct conversions between doxastic Kripke models and hypergraph models. The paper concludes with an outline of directions for future work (Section~\ref{sec:conclusion}).

\section{Preliminaries}\label{sec:preliminaries}

\begin{table}[t] 
\caption{Common axiom schemes of doxastic logic.}\label{tab:belief:axioms}
\resizebox{\linewidth}{!}{%
\begin{tabular}{c@{~}l@{\hskip .75em}l@{\hskip .75em}l}
\hline
\multicolumn{3}{l}{Axiom Scheme} & Frame Property\\
\hline
$\DB$ & $\neg\B{a}(\varphi\wedge\neg\varphi)$ & Consistency of beliefs & \emph{Seriality:} $\forall u\exists v(u\Brel{a} v)$\\
$\FourB$ & $\B{a}\varphi\rightarrow\B{a}\B{a}\varphi$ & Positive introspection & \emph{Transitivity:} $(u\Brel{a} v \andsymbol v\Brel{a} w) \Rightarrow u\Brel{a} w$\\
$\FiveB$ & $\neg\B{a}\varphi\rightarrow\B{a}\neg\B{a}\varphi$ & Negative introspection & \emph{Euclideanity:} $(u\Brel{a} v \andsymbol u\Brel{a} w) \Rightarrow v\Brel{a}w$\\
\hline
\end{tabular}}
\end{table}

Throughout this paper, we assume a fixed set $\Ag$ of $n\geq 1$ agents, denoted by the meta-variables $a$, $b$, $c$, etc. Each agent $a\in\Ag$ is assigned a countable set $\Var{a}$ of propositional letters, referred to as the \emph{local variables of $a$}. The sets of local variables are assumed to be pairwise disjoint, so we have $\Var{a}\cap\Var{b} = \emptyset$ for all $a,b\in\Ag$ with $a\neq b$. The \emph{set of all variables} is given by $\Atoms := \bigcup_{a\in\Ag}\Var{a}$. Elements of $\Atoms$ are denoted by the meta-variables $p$, $q$, $r$, etc. The \emph{language of epistemic-doxastic logic}, notation $\LKB$, is generated by the following grammar: 
\begin{equation*}
\varphi ::= p \mid \neg\varphi \mid \varphi\wedge\varphi \mid \B{a}\varphi \mid \K{a}\varphi \qquad \text{($p\in\Atoms$ and $a\in\Ag$).}
\end{equation*} 
We also put $\varphi\vee\psi := \neg(\neg\varphi\wedge\neg\psi)$, $\varphi\rightarrow\psi := \neg\varphi\vee\psi$, $\varphi\leftrightarrow\psi := (\varphi\rightarrow\psi)\wedge(\psi\rightarrow\varphi)$ as well as $\bot := p\wedge\neg p$, where $p\in\Atoms$ is an arbitrary but fixed variable. Given a formula $\varphi\in\LKB$ and an agent $a\in\Ag$, we will say that $\varphi$ is an \emph{$a$-formula}, if every propositional letter occurring in $\varphi$ is a local variable from $\Var{a}$ and every modal operator in $\varphi$ is either of the form $\B{a}$ or of the form $\K{a}$. As usual, $\B{a}\varphi$ may be read as ``Agent $a$ believes $\varphi$'' and $\K{a}\varphi$ may be read as ``Agent $a$ knows $\varphi$''. The \emph{doxastic fragment of $\LKB$} is denoted by $\LB$ and given by the grammar
\begin{equation*}
\varphi ::= p \mid \neg\varphi \mid \varphi\wedge\varphi \mid \B{a}\varphi \qquad \text{($p\in\Atoms$ and $a\in\Ag$).}
\end{equation*}
The left column of Table~\ref{tab:belief:axioms} contains three axiom schemes that are usually considered in the literature on doxastic logic (cf.\ \cite{meyer:hoek:1995,fagin:etal:2003,ditmarsch:etal:2008,humberstone:2015}). Axiom scheme $\DB$ formalizes the assumption that beliefs are \emph{consistent}, and the schemes $\FourB$ and $\FiveB$ amount to the stipulation that agents are \emph{introspective}: if an agent believes (or does not believe) a formula $\varphi$, then she also believes that she believes (or that she does not believe) $\varphi$. The normal modal logic determined by $\FourB$ and $\FiveB$ is known as $\KFourFive$ and the one determined by all three axioms is known as $\KDFourFive$. As indicated in the table, $\KFourFive$ may also be characterized as the logic of Kripke frames satisfying transitivity and Euclideanity and $\KDFourFive$ may be characterized as the logic of frames satisfying seriality, transitivity and Euclideanity \cite[Sect.~3.1]{fagin:etal:2003}.

\begin{definition}[Doxastic Kripke Model]
By a \emph{doxastic Kripke frame}, we mean a pair $F = (W,\DoxRels)$, where $W\neq\emptyset$ is a set of worlds and $\DoxRels:\Ag\rightarrow \mathcal{P}(W\times W)$ is a function that assigns, to each $a\in\Ag$, a doxastic accessibility relation $\DoxRels(a)\subseteq W\times W$. The relation $\DoxRels(a)$ will also be denoted by $\Brel{a}$, so we put ${\Brel{a}} := \DoxRels(a)$ for all $a\in\Ag$. A doxastic frame $F = (W,\DoxRels)$ is \emph{serial} (resp., \emph{transitive} or \emph{Euclidean}), if each of the relations $\Brel{a}$ is serial (resp., transitive or Euclidean). A \emph{doxastic Kripke model} is a triple $M = (W,\DoxRels,V)$, where $(W,\DoxRels)$ is a doxastic frame and $V:W\rightarrow\mathcal{P}(\Atoms)$ is a valuation. A doxastic model is \emph{serial} (resp., \emph{transitive} or \emph{Euclidean}), if the underlying frame is serial (resp., transitive or Euclidean). 
\end{definition} 

By enriching a doxastic model with epistemic accessibility relations $\Krel{a}\subseteq W\times W$ for all $a\in\Ag$, one may formalize not only the \emph{beliefs} of agents, but also their \emph{knowledge}. Here, the modalities $\K{a}$ are assumed to satisfy the usual $\SFive$-axioms, so each of the relations $\Krel{a}$ is taken to be an \emph{equivalence relation} (cf.\ \cite[Sect.~4.1]{blackburn:etal:2001}). The interplay between the beliefs and the knowledge of an agent will be governed by the three axiom schemes presented in Table~\ref{tab:interaction:axioms}. A philosophical discussion of these schemes can be found in \cite{hintikka:1962,lenzen:1978,halpern:1996,stalnaker:2006,aucher:2014,humberstone:2015}. The strong introspection axioms, $\SPI$ and $\SNI$, reflect the assumption that agents have knowledge of their beliefs: if an agent believes (or does not believe) a formula $\varphi$, then she also knows that she believes (or that she does not believe) $\varphi$. The scheme $\KIB$, on the other hand, accounts for the idea that agents always believe what they know. By combining the frame conditions determined by these schemes with the usual properties of the epistemic and doxastic accessibility relations, we now arrive at the following notion of an \emph{epistemic-doxastic model}.

\begin{table}[t]
\caption{Principles connecting knowledge and belief.}\label{tab:interaction:axioms}
\resizebox{\linewidth}{!}{%
\begin{tabular}{c@{~}l@{\hskip 1.5em}l@{\hskip 1.5em}l}
\hline
\multicolumn{3}{l}{Axiom Scheme} & Frame Property\\
\hline
$\KIB$ & $\K{a}\varphi\rightarrow\B{a}\varphi$ & Knowledge implies belief & $\Cond{1}$~$u\Brel{a} v \Rightarrow u\Krel{a} v$\\
$\SPI$ & $\B{a}\varphi\rightarrow\K{a}\B{a}\varphi$ & Strong pos.\ introspection & $\Cond{2}$~$(u\Krel{a} v \andsymbol v\Brel{a} w) \Rightarrow u\Brel{a} w$\\
$\SNI$ & $\neg\B{a}\varphi\rightarrow\K{a}\neg\B{a}\varphi$ & Strong neg.\ introspection & $\Cond{3}$~$(u\Krel{a} v \andsymbol u\Brel{a} w) \Rightarrow v\Brel{a} w$\\
\hline
\end{tabular}}
\end{table}

\begin{definition}[Epistemic-Doxastic Model]
An \emph{epistemic-doxastic frame} is a triple $(W,\DoxRels, \EpRels)$, where $(W,\DoxRels)$ is a serial, transitive and Euclidean doxastic frame and ${\EpRels}:\Ag\rightarrow\mathcal{P}(W\times W)$ is a function that assigns, to each $a\in\Ag$, an equivalence relation ${\Krel{a}} := {\EpRels}(a)$ satisfying the properties $\Cond{1}$--$\Cond{3}$ from Table~\ref{tab:interaction:axioms}. An \emph{epistemic-doxastic model} is a quadruple $(W,\DoxRels,\EpRels,V)$, where $(W,\DoxRels,\EpRels)$ is an epistemic-doxastic frame and $V:W\rightarrow\mathcal{P}(\Atoms)$ is a valuation. 
\end{definition}

Given an arbitrary binary relation $R\subseteq X\times X$ over some set $X$, we will henceforth write $R^{m}$ for the \emph{$m$-th power of $R$}, which is inductively defined by $R^{0} := \{(x,x) \mid x\in X\}$ and $R^{i+1} := R\circ R^{i}$ for all $i\geq 0$, where $\circ$ denotes the composition of relations. By the \emph{inverse} of $R$, we will mean the relation $R^{-1} := \{(x,y) \mid (y,x)\in R\}$. The \emph{symmetric closure} of $R$ is given by $\sym{R} := R\cup R^{-1}$ and the \emph{reflexive-transitive closure} of $R$ is given by $R^{*} := \bigcup_{i\geq 0} R^{i}$. The \emph{equivalence relation generated by $R$}, notation $\eqrel{R}$, is the smallest equivalence relation containing $R$. Formally, $\eqrel{R}$ can be defined as the intersection of all equivalence relations $S\subseteq X\times X$ with $R\subseteq S$. One readily sees that $\eqrel{R}$ coincides with the reflexive-transitive closure of the symmetric closure of $R$, so we have $\eqrel{R} = (\sym{R})^{*}$ (cf.\ \cite[Sect.~2.1.1]{baader:nipkow:1998}). It is straightforward to show that, in an epistemic-doxastic frame, the epistemic accessibility relation $\Krel{a}$ of an agent $a\in\Ag$ is always \emph{uniquely determined} by her doxastic accessibility relation $\Brel{a}$. In fact, $\Krel{a}$ simply amounts to the equivalence relation generated by $\Brel{a}$.  

\begin{proposition}
Let $F = (W,\DoxRels, \EpRels)$ be an epistemic-doxastic frame and let $a\in\Ag$. The epistemic accessibility relation $\Krel{a}$ coincides with $\eqrel{\Brel{a}}$, so $\Krel{a} = \eqrel{\Brel{a}}$.
\end{proposition}

\begin{proof}
For the left-to-right inclusion, let $u,v\in W$ be arbitrary and assume that $u\Krel{a} v$. By the seriality of $\Brel{a}$, there is some $w\in W$ with $v\Brel{a} w$. Since $u\Krel{a} v$ and $v\Brel{a} w$, we have $u\Brel{a} w$ by property $\Cond{2}$. Together with $v\Brel{a} w$, this yields $u\eqrel{\Brel{a}} v$. For the other inclusion, one may use induction on $m$ in order to prove $(\sym{\Brel{a}})^{m}\subseteq {\Krel{a}}$ for all $m\in\mathbb{N}$. Since $\eqrel{\Brel{a}} = (\sym{\Brel{a}})^{*}$, this establishes the claim. 
\end{proof}

So, in order to formalize both the knowledge and the beliefs of an agent, it is sufficient to know the agent's doxastic accessibility relation, provided that this relation satisfies the usual $\KDFourFive$-properties. For this reason, we henceforth restrict ourselves to \emph{doxastic models} and interpret the modalities $\K{a}$ in terms of the generated relations $\eqrel{\Brel{a}}$. The notion of satisfaction is thus defined as follows. 

\begin{definition}[Satisfaction]
Let $M = (W,\DoxRels,V)$ be a doxastic Kripke model and let $w\in W$ be a world. The \emph{satisfaction relation $\semK$} is defined as follows: 
\begin{itemize}
\item $M,w\semK p ~:\Leftrightarrow~ p\in V(w)$,
\item $M,w\semK \neg\varphi ~:\Leftrightarrow~ M,w\not\semK\varphi$,
\item $M,w\semK \varphi\wedge\psi ~:\Leftrightarrow~ \text{$M,w\semK\varphi$ and $M,w\semK\psi$}$,
\item $M,w\semK \B{a}\varphi ~:\Leftrightarrow~ \text{$M,u\semK\varphi$ for all $u\in W$ with $w\Brel{a} u$}$,
\item $M,w\semK \K{a}\varphi ~:\Leftrightarrow~ \text{$M,u\semK\varphi$ for all $u\in W$ with $w \eqrel{\Brel{a}} u$}$.
\end{itemize}
\end{definition}   

If $M,w\semK\varphi$ holds, then we say that $\varphi$ is \emph{satisfied} at $w$ in $M$. A formula $\varphi\in\LKB$ is said to be \emph{valid in a doxastic Kripke model $M$}, notation $M\semK\varphi$, if we have $M,w\semK\varphi$ for every world $w$ in $M$. And $\varphi$ is \emph{valid in a class of doxastic models $\mathcal{C}$}, notation $\mathcal{C}\semK\varphi$, if $M\semK\varphi$ holds for every model $M\in\mathcal{C}$.

\begin{definition}[Locality and Properness]\label{def:locality:properness}
A doxastic frame $F = (W,\DoxRels)$ is \emph{proper}, if for all worlds $u,v\in W$ with $u\neq v$, there exists some agent $a\in\Ag$ such that $(u,v)\notin\eqrel{\Brel{a}}$. A model is proper, if the underlying frame is proper. We say that a doxastic model $M = (W,\DoxRels,V)$ is \emph{local}, if for every $a\in\Ag$ and for all $u,v\in W$, it is the case that $(u,v)\in \sym{\Brel{a}}$ implies $V(u)\cap\Var{a} = V(v)\cap\Var{a}$. 
\end{definition}

Observe that our notion of properness is similar to the one adopted in \cite{goubault:etal:2018}: a frame is proper, if any two worlds can be distinguished by at least one agent. However, in our case, distinguishability is defined in terms of the generated equivalence relations $\eqrel{\Brel{a}}$. According to Definition~\ref{def:locality:properness}, a model is local if any two worlds that are connected to each other by an agent's doxastic accessibility relation also agree on all local variables of this agent. Roughly speaking, locality amounts to the assumption that agents are never mistaken about their own local variables: if a local variable is true at a world $w$, then the corresponding agent also believes it to be true at $w$. More precisely, one can prove the following fact.

\begin{fact}[Local Veracity of Belief]
Let $M$ be a local doxastic model, let $a\in\Ag$ and let $\varphi\in\LKB$ be an $a$-formula. If $M$ is transitive and Euclidean, then $M\semK\varphi\rightarrow\B{a}\varphi$. And if $M$ is serial, transitive and Euclidean, then $M\semK\varphi\leftrightarrow\B{a}\varphi$. 
\end{fact}

\begin{figure}[t]
\setlength\fboxsep{.75em}
\noindent\framebox[\textwidth]{\parbox{\textwidth}{
\centering\small 
\begin{tabularx}{.95\textwidth}{l@{~}l@{\hskip .5em}l@{~}X}
\multicolumn{4}{l}{\emph{Axioms}: all instances of classical tautologies as well as the following schemes:}\\[.25em]
$\KB$ & $\B{a}(\varphi\rightarrow\psi)\rightarrow (\B{a}\varphi\rightarrow\B{a}\psi)$ & $\KK$ & $\K{a}(\varphi\rightarrow\psi)\rightarrow (\K{a}\varphi\rightarrow\K{a}\psi)$\\
$\DB$ & $\neg\B{a}(\varphi\wedge\neg\varphi)$ & $\FourB$ & $\B{a}\varphi\rightarrow\B{a}\B{a}\varphi$\\
$\FiveB$ &  $\neg\B{a}\varphi\rightarrow\B{a}\neg\B{a}\varphi$ & $\TK$ & $\K{a}\varphi\rightarrow\varphi$ \\
$\FourK$ & $\K{a}\varphi\rightarrow\K{a}\K{a}\varphi$ & $\FiveK$ & $\neg\K{a}\varphi\rightarrow\K{a}\neg\K{a}\varphi$\\
$\SPI$ & $\B{a}\varphi\rightarrow\K{a}\B{a}\varphi$ & $\SNI$ & $\neg\B{a}\varphi\rightarrow\K{a}\neg\B{a}\varphi$\\
$\KIB$ & $\K{a}\varphi\rightarrow\B{a}\varphi$ & $\Loc$ & $(p\rightarrow\B{a} p)\wedge(\neg p\rightarrow\B{a} \neg p)$~for $p\in \Var{a}$\\[.5em]
\multicolumn{4}{l}{-- \textit{Modus ponens}: from $\varphi$ and $\varphi\rightarrow\psi$, infer $\psi$.} \\
\multicolumn{4}{l}{-- \textit{Necessitation}: from $\varphi$, infer $\K{a}\varphi$.}
\end{tabularx}}}
\caption{The Hilbert-style system $\EDL$.}\label{fig:hilbert:system:EDL}
\end{figure}

We will henceforth write $\KTEclass$ for the class of all local and proper doxastic Kripke models that are \emph{transitive} and \emph{Euclidean}, and we write $\KSTEclass$ for the class of all local and proper doxastic models that are \emph{serial}, \emph{transitive} and \emph{Euclidean} (here and in the following, the subscript $n$ is used to emphasize that we are considering models and logics with $n$ agents). A Hilbert-style axiomatization of epistemic-doxastic logic based on local models is presented in Figure~\ref{fig:hilbert:system:EDL}. We will refer to this axiomatization as $\EDL$. As can be seen, $\EDL$ comprises the usual $\KDFourFive$-axioms for belief, the $\SFive$-axioms for knowledge and the three interaction principles from Table~\ref{tab:interaction:axioms} (observe that the schemes $\FourB$ and $\FiveB$ are actually redundant, since they are already derivable from $\SPI$, $\SNI$ and $\KIB$; we decided to include them here only for the sake of comprehensiveness). The additional axiom $\Loc$ reflects the locality constraint for doxastic models. A \emph{proof} in our system is a finite sequence of formulas $\varphi_1,\ldots,\varphi_m$, where each $\varphi_i$ is either an axiom of $\EDL$ or the result of applying modus ponens or necessitation to formulas occurring earlier in the sequence. We will write $\EDL\prov\varphi$ and say that $\varphi$ is \emph{provable} in our system, if there exists a proof in $\EDL$ that ends with $\varphi$.

\begin{theorem}[Soundness and Completeness]\label{th:EDL:complete:wrt:kripke:models}
The proof system $\EDL$ is sound and complete with respect to the class $\KSTEclass$ of all local and proper doxastic Kripke models that are serial, transitive and Euclidean. That is, for every formula $\varphi\in\LKB$, we have: $\EDL\prov\varphi$ iff $\KSTEclass\semK\varphi$.
\end{theorem} 

As usual, the soundness is established by induction on the length of a proof in $\EDL$. For the completeness part, one may use a standard canonical model construction. The proof works in essentially the same way as in \cite[Theorems~3.1.3 and 3.1.5]{fagin:etal:2003}. We omit the details here, since Theorem~\ref{th:EDL:complete:wrt:kripke:models} will not play any significant role in the further course of our investigation. In order to formalize \emph{only} the beliefs of agents (but not their knowledge), one may omit all axioms involving the $\K{}$-modalities and restrict the remaining schemes to formulas from the doxastic fragment $\LB$. The resulting proof system, henceforth denoted by $\LocKDFourFive$, is presented in Figure~\ref{fig:hilbert:system:locKDFourFive} and can be seen as a local variant of multi-agent $\KDFourFive$. A weaker doxastic logic may be obtained from $\LocKDFourFive$ by also omitting axiom scheme $\DB$. We will refer to this weaker system as $\LocKFourFive$. It is easy to show that $\LocKFourFive$ is sound and complete with respect to the class $\KTEclass$ and that $\LocKDFourFive$ is sound and complete with respect to the class $\KSTEclass$. That is, for every formula $\varphi$ from the fragment $\LB$, we have $\LocKFourFive\prov\varphi\Leftrightarrow\KTEclass\semK\varphi$ as well as $\LocKDFourFive\prov\varphi\Leftrightarrow\KSTEclass\semK\varphi$. The proof is again standard and therefore omitted.

\begin{figure}[t]
\setlength\fboxsep{.75em}
\noindent\framebox[\textwidth]{\parbox{\textwidth}{
\centering\small
\begin{tabularx}{.95\textwidth}{l@{~}X@{\hskip 3em}l@{~}X}
\multicolumn{4}{l}{\emph{Axioms}: all instances of classical tautologies as well as the following schemes:}\\[.25em]
$\KB$ & $\B{a}(\varphi\rightarrow\psi)\rightarrow (\B{a}\varphi\rightarrow\B{a}\psi)$ & $\DB$ &  $\neg\B{a}(\varphi\wedge\neg\varphi)$\\
$\FourB$ & $\B{a}\varphi\rightarrow\B{a}\B{a}\varphi$ & $\FiveB$ & $\neg\B{a}\varphi\rightarrow\B{a}\neg\B{a}\varphi$\\
$\Loc$ & \multicolumn{3}{@{}l}{$(p\rightarrow\B{a} p)\wedge(\neg p\rightarrow\B{a} \neg p)$~for $p\in \mathsf{Var}_a$}\\[.5em]
\multicolumn{4}{l}{-- \textit{Modus ponens}: from $\varphi$ and $\varphi\rightarrow\psi$, infer $\psi$.} \\
\multicolumn{4}{l}{-- \textit{Necessitation}: from $\varphi$, infer $\B{a}\varphi$.}
\end{tabularx}
}}
\caption{The Hilbert-style system $\LocKDFourFive$.}\label{fig:hilbert:system:locKDFourFive}
\end{figure}

\section{Hypergraph Semantics for Doxastic Logics}\label{sec:hypergraph:semantics}

We now present an alternative semantics for doxastic logic. In contrast to Kripke models, the primary objects of our semantics are not taken to be \emph{possible worlds} (global states), but \emph{local states} of the agents (representing their ``perspectives'' on the worlds). Our models will be built up from \emph{directed hypergraphs}, a generalization of ordinary directed graphs in which edges are allowed to connect more than two vertices \cite{gallo:etal:1993}. The vertices of a directed hypergraph are interpreted as local states and the edges are interpreted as global states. Hypergraphs may also be conceived as a generalization of \emph{simplicial complexes}, which are known to provide a semantics for various systems of epistemic logic \cite{ledent:2019,goubault:etal:2021,ditmarsch:etal:2021,goubault:etal:2022,ditmarsch:etal:2022,randrianomentsoa:etal:2023}. Our semantics preserves the characteristic features of simplicial models (namely, their geometric nature and the precedence of local states over global states), while being able to capture not only the \emph{knowledge}, but also the \emph{beliefs} of agents. 

We start by introducing some basic notions. By an \emph{undirected hyperedge} over a set of vertices $\Vertices$, we simply mean a finite subset $\edge\subseteq \Vertices$. The vertices contained in an undirected hyperedge are thought of as being \emph{adjacent} to each other by means of this hyperedge \cite[Ch.~17]{berge:1973}. A \emph{directed hyperedge} over $\Vertices$, on the other hand, is defined to be an ordered pair $\edge=(X,Y)$, where $X,Y\subseteq \Vertices$ are finite and disjoint sets of vertices \cite{gallo:etal:1993}. Note that, in particular, $X$ and $Y$ are also allowed to be empty. The set $X$ is said to be the \emph{tail} of $\edge$ and the set $Y$ is said to be the \emph{head} of $\edge$. In what follows, the tail and the head of a directed hyperedge $\edge$ will be denoted by $\tail{\edge}$ and $\head{\edge}$, respectively. Given a directed hyperedge $\edge$, we also write $\undir{\edge} := \tail{\edge}\cup\head{\edge}$ for the \emph{undirected hyperedge induced by $\edge$}.

\begin{definition}[Directed and Undirected Hypergraphs]
A \emph{directed hypergraph} is defined to be a pair $G=(\Vertices,\Edges)$, where $\Vertices$ is a non-empty set of vertices and $\Edges\subseteq \mathcal{P}(\Vertices)\times\mathcal{P}(\Vertices)$ is a set of directed hyperedges. The notion of an \emph{undirected hypergraph} is defined in the same way, except that, in this case, $\Edges$ is taken to be a set of undirected hyperedges, so $\Edges\subseteq\mathcal{P}(\Vertices)$. 
\end{definition} 

We will write $\Vertices(G)$ for the set of \emph{vertices} and $\Edges(G)$ for the set of \emph{hyperedges} of a hypergraph $G$. A directed hypergraph $G$ is said to be \emph{simple}, if for all hyperedges $\edge_1,\edge_2\in\Edges(G)$, it is the case that $\edge_1\neq\edge_2$ implies $\undir{\edge_1}\not\subseteq\undir{\edge_2}$. The \emph{rank} of a directed hypergraph $G$, notation $\rk{G}$, is defined to be the maximum number of vertices occurring in the hyperedges of $G$, so we put $\rk{G} := \max_{\edge\in\Edges(G)} |\undir{\edge}|$ (if there is no finite bound on the size of hyperedges in $G$, then we put $\rk{G} := \infty$). We will say that a hypergraph $G$ is \emph{$k$-uniform} for some $k\in\mathbb{N}$, if every hyperedge of $G$ contains exactly $k$ vertices, i.e., if we have $|\undir{\edge}| = k$ for all $\edge\in\Edges(G)$. And $G$ is said to be \emph{tail-complete}, if every vertex of $G$ belongs to the tail of at least one hyperedge, so $\Vertices(G) = \bigcup_{\edge\in\Edges(G)}\tail{\edge}$.

Note that, by ignoring the directions of hyperedges, any $k$-uniform hypergraph can also be given a natural geometric interpretation. In particular, a $2$-uni\-form graph corresponds to a set of \emph{line segments}, a $3$-uni\-form graph to a set of \emph{triangles}, a $4$-uni\-form graph to a set of \emph{tetrahedrons}, etc. More generally, every simple $k$-uniform hypergraph may be conceived as a collection of $(k-1)$-dimensional \emph{simplices} glued together at a proper subset of their vertices.

\begin{remark}
An \emph{abstract simplicial complex} $C$ over a set of vertices $\Vertices$ is a set of non-empty finite subsets of $\Vertices$ that is closed under set inclusion (i.e., if $X\in C$ and $\emptyset\neq Y\subseteq X$, then $Y\in C$) and that also contains all singleton subsets of $\Vertices$. The elements of a simplicial complex $C$ are referred to as \emph{simplices}. A simplex $X\in C$ is said to be a \emph{facet} of $C$, if $X$ is not a proper subset of any other simplex in $C$. The \emph{dimension of a simplex} $X\in C$ is given by $\dim(X) := |X|-1$. A simplicial complex $C$ is \emph{pure}, if all facets of $C$ have the same dimension \cite[Sect.~3.2]{herlihy:etal:2014}. Obviously, every simplicial complex may be conceived as an undirected hypergraph that contains only non-empty hyperedges and that is also \emph{downward closed} (i.e., if $\edge$ is a hyperedge in the graph, then so is every non-empty subset of $\edge$). The \emph{simplicial complex induced by a directed hypergraph $G$} is given by $C_G := \{X \mid \text{$\exists \edge\in\Edges(G)$ with $X\subseteq\undir{\edge}$ and $X\neq\emptyset$}\}$. The facets of $C_G$ correspond to those hyperedges $\edge\in\Edges(G)$ for which $\undir{\edge}$ is maximal with respect to inclusion. It is easy to see that, if $G$ is simple, then there is also a \emph{one-to-one correspondence} between the facets of $C_G$ and the hyperedges in $G$. And if $G$ is $k$-uniform, then every facet of $C_G$ is of dimension $k-1$, so $C_G$ is a \emph{pure} simplicial complex.
\end{remark} 

By a \emph{coloring} of a directed hypergraph $G$, we mean a function $\colmap:\Vertices(G)\rightarrow\Ag$ that assigns, to each vertex $u\in \Vertices(G)$, an agent $\colmap(u)\in\Ag$, in such a way that distinct vertices belonging to the same hyperedge are assigned different colors under $\colmap$. More precisely, for every $\edge\in\Edges(G)$ and for all $u,v\in \undir{\edge}$ with $u\neq v$, we should have $\colmap(u)\neq \colmap(v)$. If such a coloring exists for a directed hypergraph $G$, then $G$ is said to be \emph{colorable}. Clearly, every colorable hypergraph $G$ must satisfy $\rk{G}\leq |\Ag|$. A \emph{chromatic hypergraph} is a pair $(G,\colmap)$, where $G$ is a colorable directed hypergraph and $\colmap:\Vertices(G)\rightarrow\Ag$ is a coloring of $G$. We say that a chromatic hypergraph $(G,\colmap)$ is \emph{simple} (respectively, \emph{$k$-uni\-form} or \emph{tail-complete}), if the underlying hypergraph $G$ is simple (respectively, $k$-uniform or tail-com\-plete). Moreover, for any subset $S\subseteq\Vertices(G)$, we put $\colmap(S) := \{\colmap(u) \mid u\in S\}$.

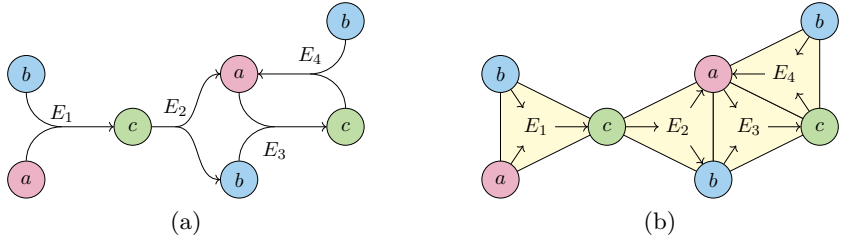
\begin{figure}[t]
\centering%
\subfloat[]{\label{subfig:uniform:graph:a}%
\scalebox{.7}{%
\begin{tikzpicture}[minimum size=.7cm]
\node[circle,draw,fill=myred] (a1) at (0,0) {$a$};
\node[circle,draw,fill=myblue] (b1) at (0,2) {$b$};
\node[circle,draw,fill=mygreen] (c1) at (2,1) {$c$};
\node[circle,draw,fill=myred] (a2) at (4,2) {$a$};
\node[circle,draw,fill=myblue] (b2) at (4,0) {$b$};
\node[circle,draw,fill=mygreen] (c2) at (6,1) {$c$};
\node[circle,draw,fill=myblue] (b3) at (6,3) {$b$};

\coordinate[label={[shift={(0,-0.1)}]$\edge_1$}] (E1) at (barycentric cs:a1=1,b1=1,c1=1);
\coordinate[label={[shift={(0,0)}]$\edge_2$}] (E2) at (barycentric cs:a2=1,b2=1,c1=3);
\coordinate[label={[below,shift={(0,-0.1)}]$\edge_3$}] (E3) at (barycentric cs:a2=1,b2=1,c2=1);
\coordinate[label={[shift={(0,0)}]$\edge_4$}] (E4) at (barycentric cs:a2=1,b3=1,c2=1);

\draw[->] (a1) to[out=90,in=180] (E1) -- (c1);
\draw (b1) to[out=-90,in=180] (E1);

\draw[->] (c1) -- (E2) to[out=0,in=180] (a2);
\draw[->] (E2) to[out=0,in=180] (b2);

\draw[->] (a2) to[out=-90,in=180] (E3) -- (c2);
\draw (b2) to[out=90,in=180] (E3);

\draw[->] (b3) to[out=270,in=0] (E4) -- (a2);
\draw (c2) to[out=90,in=0] (E4);
\end{tikzpicture}}}%
\hspace*{4em}%
\subfloat[]{\label{subfig:uniform:graph:b}%
\scalebox{.7}{%
\begin{tikzpicture}[minimum size=.7cm]
\node[circle,draw,fill=myred] (a1) at (0,0) {$a$};
\node[circle,draw,fill=myblue] (b1) at (0,2) {$b$};
\node[circle,draw,fill=mygreen] (c1) at (2,1) {$c$};
\node[circle,draw,fill=myred] (a2) at (4,2) {$a$};
\node[circle,draw,fill=myblue] (b2) at (4,0) {$b$};
\node[circle,draw,fill=mygreen] (c2) at (6,1) {$c$};
\node[circle,draw,fill=myblue] (b3) at (6,3) {$b$};

\node (E1) at (barycentric cs:a1=1,b1=1,c1=1) {$\edge_1$};
\node (E2) at (barycentric cs:a2=1,b2=1,c1=1) {$\edge_2$};
\node (E3) at (barycentric cs:a2=1,b2=1,c2=1) {$\edge_3$};
\node (E4) at (barycentric cs:a2=1,b3=1,c2=1) {$\edge_4$};

\begin{scope}[on background layer]
\filldraw[fill=mybgcolor] (a1.center)--(b1.center)--(c1.center)--cycle;
\filldraw[fill=mybgcolor] (a2.center)--(b2.center)--(c1.center)--cycle;
\filldraw[fill=mybgcolor] (a2.center)--(b2.center)--(c2.center)--cycle;
\filldraw[fill=mybgcolor] (a2.center)--(b3.center)--(c2.center)--cycle;
\end{scope} 

\draw[->] (a1) edge (E1) (b1) edge (E1) (E1) edge (c1);
\draw[->] (c1) edge (E2) (E2) edge (a2) (E2) edge (b2);
\draw[->] (a2) edge (E3) (b2) edge (E3) (E3) edge (c2);
\draw[->] (E4) edge (a2) (b3) edge (E4) (c2) edge (E4);
\end{tikzpicture}}}%
\caption{Two visualizations of a simple $3$-uniform chromatic hypergraph.}\label{fig:uniform:graph}
\end{figure}

\begin{example}
Figure~\ref{fig:uniform:graph} comprises two different illustrations of a simple $3$-uniform chromatic hypergraph. The illustration in Figure~\ref{subfig:uniform:graph:a} uses a common method for the visualization of directed hypergraphs. In this illustration, hyperedges are drawn as ``multi-arrows'', i.e., arrows that may have an arbitrary number of source and target nodes \cite{gallo:etal:1993}. The illustration in Figure~\ref{subfig:uniform:graph:b} emphasizes the geometric nature of $k$-uniform hypergraphs and their connection to pure simplicial complexes. There, hyperedges are represented as triangles (possibly intersecting each other at some of their vertices). Note that, in the graph shown, every vertex is contained in the tail of some hyperedge, so this graph is also tail-complete. 
\end{example}
 
\begin{definition}[Hypergraph Model]\label{def:hypergraph:model}
A \emph{hypergraph model} is a triple $M=(G,\colmap,\ell)$, where $(G,\colmap)$ is a chromatic directed hypergraph and $\ell:\Vertices(G)\rightarrow \mathcal{P}(\Atoms)$ is a valuation function such that, for all $u\in\Vertices(G)$, if $\colmap(u) = a$, then $\ell(u)\subseteq\Var{a}$.
\end{definition}

Given a hypergraph model $M=(G,\colmap,\ell)$ and an edge $\edge\in\Edges(G)$, we will also write $\ell(\edge)$ for the set of propositional variables defined by $\ell(\edge) := \bigcup_{u\in\undir{\edge}} \ell(u)$. 

\begin{definition}[Simplicity, Uniformity, Tail-Completeness]
Let $n = |\Ag|$. A hypergraph model $M=(G,\colmap,\ell)$ is \emph{simple} (resp., \emph{$n$-uniform} or \emph{tail-complete}), if the hypergraph $G$ is simple (resp., $n$-uniform or tail-complete). We write $\HSUclass$ for the class of all simple and $n$-uniform hypergraph models, and we write $\HSUTclass$ for the class of all simple, $n$-uniform and tail-complete hypergraph models.
\end{definition}

Intuitively, every $a$-colored vertex of a directed chromatic hypergraph represents a \emph{local state} of agent $a\in\Ag$. The directed hyperedges of the graph are taken to represent \emph{possible worlds} or \emph{global states}. Two global states $\edge_1$ and $\edge_2$ are \emph{epistemically indistinguishable} for an agent $a\in\Ag$, if there is an $a$-col\-ored vertex that belongs to both $\undir{\edge_1}$ and $\undir{\edge_2}$. And a global state $\edge_2$ is considered to be \emph{doxastically possible} by an agent $a\in\Ag$ at a global state $\edge_1$, if there is an $a$-colored vertex that belongs to $\undir{\edge_1}$ and to the \emph{tail} $\tail{\edge_2}$ of the state $\edge_2$. We thus adopt the following notions of \emph{accessibility} between global states. 

\begin{definition}[Accessibility Relations]
Let $(G,\colmap)$ be a chromatic hypergraph. For every agent $a\in\Ag$, we define a \emph{doxastic accessibility relation} $\GBrel{a}$ and an \emph{epistemic accessibility relation} $\GKrel{a}$ as follows, for all $\edge_1,\edge_2\in\Edges(G)$: 
\begin{itemize}
\item $\edge_1\GBrel{a}\edge_2 ~:\Leftrightarrow~ a\in\colmap(\undir{\edge_1}\cap\tail{\edge_2})$,
\item $\edge_1\GKrel{a}\edge_2 ~:\Leftrightarrow~ a\in\colmap(\undir{\edge_1}\cap\undir{\edge_2})$.
\end{itemize} 
\end{definition} 

So, for example, in the graph from Figure~\ref{fig:uniform:graph}, we have $\edge_2\GBrel{a}\edge_3$ and $\edge_3\GBrel{a}\edge_3$, but neither $\edge_2\GBrel{a}\edge_2$ nor $\edge_3\GBrel{a}\edge_2$. Observe that $\GKrel{a}$ essentially amounts to the usual indistinguishability relation between facets in simplicial models \cite{goubault:etal:2021,ditmarsch:etal:2022}. 

\begin{proposition}\label{prop:properties:defined:relations}
Let $a\in\Ag$ be an agent. In every $n$-uniform chromatic hy\-per\-graph $(G,\colmap)$, the accessibility relations $\GBrel{a}$ and $\GKrel{a}$ have the following properties:   
\begin{enumerate}
\item\label{prop:properties:defined:relations:i} $\GBrel{a}$ is transitive and Euclidean. If $G$ is tail-complete, then $\GBrel{a}$ is also serial.
\item\label{prop:properties:defined:relations:ii} $\GKrel{a}$ is reflexive, symmetric and transitive. 
\end{enumerate} 
\end{proposition}

\begin{proof}
We only prove part~(\ref{prop:properties:defined:relations:i}). Let $(G,\colmap)$ be an $n$-uniform chro\-mat\-ic hypergraph and let $a\in\Ag$. For the transitivity of $\GBrel{a}$, assume that $\edge_1\GBrel{a}\edge_2$ and $\edge_2\GBrel{a}\edge_3$. Then, clearly, there is a vertex $u\in\Vertices(G)$ with $\colmap(u) = a$ such that $u\in \undir{\edge_1}\cap\tail{\edge_2}$ and $u\in \undir{\edge_2}\cap\tail{\edge_3}$. Hence, this vertex also satisfies $u\in \undir{\edge_1}\cap\tail{\edge_3}$, so we have $\edge_1\GBrel{a}\edge_3$, as desired. The Euclideanity of $\GBrel{a}$ can be established in a similar way. For the second half of part~(\ref{prop:properties:defined:relations:i}), suppose now that $G$ is also tail-com\-plete and let $\edge\in\Edges(G)$ be arbitrary. Since $G$ is $n$-uniform, there must be some $u\in\undir{\edge}$ with $\colmap(u) = a$. And by the tail-com\-plete\-ness of $G$, we have $u\in\tail{\edge'}$ for some $\edge'\in\Edges(G)$, so it follows $\edge\GBrel{a}\edge'$. Therefore, $\GBrel{a}$ is serial. 
\end{proof}

The satisfaction relation between global states and formulas can be defined in essentially the same way as in Kripke semantics. However, the modalities $\B{a}$ and $\K{a}$ are now interpreted in terms of the defined relations $\GBrel{a}$ and $\GKrel{a}$, respectively.  

\begin{definition}[Satisfaction]
Let $M=(G,\colmap,\ell)$ be a hypergraph model and let $\edge\in\Edges(G)$ be a hyperedge. The \emph{satisfaction relation $\semH$} is defined as follows: 
\begin{itemize}
\item $M,\edge\semH p ~:\Leftrightarrow~ p\in\ell(\edge)$,
\item $M,\edge\semH \neg\varphi ~:\Leftrightarrow~ M,\edge\not\semH\varphi$,
\item $M,\edge\semH \varphi\wedge\psi ~:\Leftrightarrow~ \text{$M,\edge\semH\varphi$ and $M,\edge\semH\psi$}$,
\item $M,\edge\semH \B{a}\varphi ~:\Leftrightarrow~ \text{$M,\edge'\semH\varphi$ for all $\edge'\in\Edges(G)$ with $\edge\GBrel{a}\edge'$}$,
\item $M,\edge\semH \K{a}\varphi ~:\Leftrightarrow~ \text{$M,\edge'\semH\varphi$ for all $\edge'\in\Edges(G)$ with $\edge\GKrel{a}\edge'$}$.
\end{itemize}
\end{definition}

If $M,\edge\semH\varphi$ holds, then we say that $\varphi$ is \emph{satisfied} by the hyperedge $\edge$ in $M$. A formula $\varphi$ is said to be \emph{valid in a hypergraph model $M$}, notation $M\semH\varphi$, if we have $M,\edge\semH\varphi$ for every hyperedge $\edge$ in $M$. And $\varphi$ is \emph{valid in a class $\mathcal{C}$ of hypergraph models}, notation $\mathcal{C}\semH\varphi$, if $M\semH\varphi$ holds for all models $M\in\mathcal{C}$.

\section{Soundness and Completeness}\label{sec:completeness}

We will now prove that hypergraph models indeed provide an appropriate semantics for the doxastic logics introduced in Section~\ref{sec:preliminaries}. In particular, we will see that both the epis\-temic-dox\-as\-tic system $\EDL$ and the pure doxastic logic $\LocKDFourFive$ are sound and complete with respect to the class $\HSUTclass$ of all simple, $n$-uni\-form and tail-complete hypergraph models. Moreover, the weak doxastic logic $\LocKFourFive$ is sound and complete with respect to the class $\HSUclass$ of all simple and $n$-uni\-form hypergraph models. It is worth noting that our completeness proof will not depend on the existence of a corresponding completeness result for Kripke models, as it is sometimes the case in the literature on simplicial semantics for epistemic logics (see, e.g., \cite[Sect.~3.2.3]{ledent:2019} or \cite[Sect.~2.2]{goubault:etal:2021}). Instead, we will establish the completeness of our calculi directly by means of \emph{canonical hypergraph models}. A similar canonical model construction (for a semantics based on impure simplicial complexes) was recently presented in \cite{randrianomentsoa:etal:2023}. For the sake of brevity, we will explain the details of our completeness proof only for the epis\-temic-dox\-as\-tic system $\EDL$. However, essentially the same argument can, \emph{mutatis mutandis}, be used to establish the completeness of the systems $\LocKFourFive$ and $\LocKDFourFive$. 

We start by introducing some terminology. Given a set of formulas $\Gamma\subseteq\LKB$, we will write $\Gamma\prov\varphi$ and say that \emph{$\varphi$ is provable from $\Gamma$ in $\EDL$}, if there exists a finite subset $\Delta\subseteq\Gamma$ such that $\EDL\prov\bigwedge\Delta\rightarrow\varphi$, where $\bigwedge\Delta$ stands for the conjunction of the formulas in $\Delta$. A set of formulas $\Gamma$ is \emph{consistent} with respect to $\EDL$, if we have $\Gamma\notprov\bot$. And $\Gamma$ is \emph{maximally consistent}, if it is consistent and no proper extension $\Delta\supsetneq\Gamma$ is also consistent. The set of all maximally $\EDL$-consistent sets will be denoted by $\MCS$. The following properties of maximally consistent sets can be verified in the usual way (see, e.g., \cite[Lemma~3.1.2]{fagin:etal:2003}).

\begin{lemma}[Properties of Maximally Consistent Sets]\label{lem:properties:mcs}
Let $\Gamma\in\MCS$.
\begin{enumerate}
\item\label{lem:properties:mcs:i} For every formula $\varphi$, either $\varphi\in\Gamma$ or $\neg\varphi\in\Gamma$,
\item\label{lem:properties:mcs:ii} if $\EDL\prov\varphi$, then $\varphi\in\Gamma$,
\item\label{lem:properties:mcs:iii} if $\varphi\in\Gamma$ and $(\varphi\rightarrow\psi)\in\Gamma$, then $\psi\in\Gamma$,
\item\label{lem:properties:mcs:iv} $(\varphi\wedge\psi)\in\Gamma$ iff $\varphi\in\Gamma$ and $\psi\in\Gamma$.
\end{enumerate}
\end{lemma}

Note that, in particular, every maximally $\EDL$-consistent set must include all instances of the $\KDFourFive$-axioms for the belief modalities $\B{a}$ and all instances of the $\SFive$-axioms for the knowledge modalities $\K{a}$ (this follows directly from part~(\ref{lem:properties:mcs:ii}) of Lemma~\ref{lem:properties:mcs}). Furthermore, it is not difficult to show that every $\EDL$-consistent set can be extended to a maximally $\EDL$-consistent set. The proof works in the same way as in \cite[Lemma~4.17]{blackburn:etal:2001} and \cite[Lemma~3.1.2]{fagin:etal:2003}.

\begin{lemma}[Lindenbaum's Lemma]\label{lem:lindenbaum}
For every consistent set of formulas $\Gamma$, there exists a maximally consistent set $\Delta\in\MCS$ such that $\Gamma\subseteq\Delta$.  
\end{lemma} 

Let us now turn to the construction of our canonical hypergraph model for $\EDL$. The basic idea is as follows: for every maximally consistent set $\Gamma$, our model will contain exactly one directed hyperedge $\edge_{\Gamma}$. Intuitively, this hyperedge $\edge_{\Gamma}$ represents the global state in which all formulas from $\Gamma$ are true and all other formulas are false. Each of the hyperedges $\edge_{\Gamma}$ will be composed of $n$ vertices (one for each agent) representing the local states of the agents in $\edge_{\Gamma}$. The vertex associated with an agent $a\in\Ag$ in a global state $\edge_{\Gamma}$ is itself taken to be a set of formulas $\locstate{\Gamma}{a}$ encoding two types of information: the values of the agent's local variables and the beliefs of the agent in $\edge_{\Gamma}$. Such a vertex $\locstate{\Gamma}{a}$ is stipulated to be in the tail of $\edge_{\Gamma}$, if the corresponding beliefs of the agent are consistent with the global state represented by $\edge_{\Gamma}$. Otherwise, $\locstate{\Gamma}{a}$ is taken to be in the head of $\edge_{\Gamma}$. The following notation will be useful for our purposes.     

\begin{definition}
Let $\Gamma$ be a set of formulas and let $a\in\Ag$ be an agent. We define: 
\begin{itemize}
\item $\modform{\Gamma}{\B{a}} := \{\varphi \mid \B{a}\varphi\in\Gamma\}$ and $\modform{\Gamma}{\K{a}} := \{\varphi \mid \K{a}\varphi\in\Gamma\}$,
\item $\Vform{a}{\Gamma} := \Gamma\cap\Var{a}$ and $\Bform{a}{\Gamma} := \{\B{a}\varphi \mid \B{a}\varphi\in\Gamma\}$,
\item $\locstate{\Gamma}{a} := \Vform{a}{\Gamma}\cup\Bform{a}{\Gamma}$,
\item $\Tset{\Gamma} := \{\locstate{\Gamma}{b} \mid \text{$b\in\Ag$ and $\modform{\Gamma}{\B{b}}\subseteq\Gamma$}\}$ and $\Hset{\Gamma} := \{\locstate{\Gamma}{b} \mid \text{$b\in\Ag$ and $\modform{\Gamma}{\B{b}}\not\subseteq\Gamma$}\}$.
\end{itemize}
\end{definition}

Our canonical hypergraph model for $\EDL$ can now be described as follows.  

\begin{definition}[Canonical Hypergraph Model]
The \emph{canonical hypergraph model} for $\EDL$ is the triple $\Mc = (\Gc,\COLc,\VALc)$ defined in the following way:     
\begin{itemize}
\item $\Gc = (\Vertices,\Edges)$ with $\Vertices = \{\locstate{\Gamma}{a} \mid a\in\Ag, \Gamma\in\MCS\}$, $\Edges = \{(\Tset{\Gamma},\Hset{\Gamma}) \mid \Gamma\in\MCS\}$, 
\item $\COLc:\Vertices(\Gc)\rightarrow\Ag$ is the coloring defined by $\COLc(\locstate{\Gamma}{a}) := a$,
\item $\VALc:\Vertices(\Gc)\rightarrow \mathcal{P}(\Atoms)$ is the valuation given by $\VALc(\locstate{\Gamma}{a}) := \Gamma\cap\Var{a}$.
\end{itemize}
For any $\Gamma\in\MCS$, we also define $\edge_{\Gamma} := (\Tset{\Gamma},\Hset{\Gamma})$.
\end{definition}

It is easy to see that $\Mc$ is indeed a well-defined hypergraph model. First of all, by construction, every hyperedge $\edge_{\Gamma}$ in $\Mc$ contains exactly one local state $\locstate{\Gamma}{a}$ for each agent $a\in\Ag$, so the function $\COLc$ cannot assign the same color to more than one vertex within an edge. Hence, $\COLc$ is in fact a \emph{coloring} of $\Gc$. Furthermore, if two maximally consistent sets $\Gamma,\Delta\in\MCS$ satisfy $\locstate{\Gamma}{a} = \locstate{\Delta}{a}$ for some agent $a\in\Ag$, then these sets must also agree on all local variables of this agent, so $\VALc$ is a well-defined function from vertices to sets of local variables. 

In order to obtain the desired completeness result, we must now essentially prove two things: on the one hand, we need to show that our canonical model for $\EDL$ is in fact simple, $n$-uniform and tail-complete. On the other hand, we need to establish a suitable variant of the \emph{truth lemma}: a formula is satisfied by a hyperedge $\edge_{\Gamma}$ in the model $\Mc$, just in case this formula is contained in the corresponding set $\Gamma$. To this end, we first prove the following facts. 

\begin{lemma}\label{lem:equality:of:mcs}
Let $\Gamma,\Delta\in\MCS$. If $\locstate{\Gamma}{a} = \locstate{\Delta}{a}$ for all $a\in\Ag$, then $\Gamma = \Delta$.     
\end{lemma}

\begin{proof}
Let $\Gamma,\Delta\in\MCS$ and suppose that $\locstate{\Gamma}{a} = \locstate{\Delta}{a}$ for all $a\in\Ag$. Then, in particular, we must have $\Vform{a}{\Gamma} = \Vform{a}{\Delta}$ and $\Bform{a}{\Gamma} = \Bform{a}{\Delta}$ for each $a\in\Ag$. Let now $\varphi\in\LKB$ be arbitrary. Using induction on $\varphi$, we prove that $\varphi\in\Gamma\Leftrightarrow\varphi\in\Delta$.

If $\varphi$ is a propositional letter or a formula of the form $\B{a}\psi$, then the claim follows directly from the fact that $\Vform{a}{\Gamma} = \Vform{a}{\Delta}$ and $\Bform{a}{\Gamma} = \Bform{a}{\Delta}$ for all $a\in\Ag$. In the cases for $\neg$ and $\wedge$, we simply apply Lemma~\ref{lem:properties:mcs} and the induction hypothesis. Suppose now that $\varphi = \K{a}\psi$ for some $a\in\Ag$ and some $\psi\in\LKB$. For the left-to-right direction, assume $\K{a}\psi\in\Gamma$. By axioms $\FourK$ and $\KIB$, we then have $\B{a}\K{a}\psi\in\Gamma$. Together with $\Bform{a}{\Gamma} = \Bform{a}{\Delta}$, this yields $\B{a}\K{a}\psi\in\Delta$. Suppose for a contradiction that $\K{a}\psi\notin\Delta$. By Lemma~\ref{lem:properties:mcs}, this implies $\neg\K{a}\psi\in\Delta$, so we have $\B{a}\neg\K{a}\psi\in\Delta$ by axioms $\FiveK$ and $\KIB$. From $\B{a}\K{a}\psi\in\Delta$ and $\B{a}\neg\K{a}\psi\in\Delta$, we obtain $\B{a}(\K{a}\psi\wedge\neg\K{a}\psi)\in\Delta$. But by axiom $\DB$, we also have $\neg\B{a}(\K{a}\psi\wedge\neg\K{a}\psi)\in\Delta$, which is a contradiction to the consistency of $\Delta$. Hence, $\K{a}\psi\in\Delta$. The right-to-left direction is established in the same way.

This concludes the induction. Since $\varphi$ was arbitrary, it follows $\Gamma = \Delta$. 
\end{proof}

\begin{lemma}\label{lem:equality:local:states}
Let $\Gamma,\Delta\in\MCS$ be maximally consistent and let $a\in\Ag$ be an agent. 
\begin{enumerate}
\item\label{lem:equality:local:states:i} If $\modform{\Gamma}{\B{a}}\subseteq\Delta$, then $\locstate{\Gamma}{a} = \locstate{\Delta}{a}$.
\item\label{lem:equality:local:states:ii} If $\modform{\Gamma}{\K{a}}\subseteq\Delta$, then $\locstate{\Gamma}{a} = \locstate{\Delta}{a}$.
\end{enumerate}
\end{lemma}

\begin{proof}
We only prove part~(\ref{lem:equality:local:states:i}). Let $\Gamma,\Delta\in\MCS$ be maximally consistent, let $a\in\Ag$ be an agent and suppose that $\modform{\Gamma}{\B{a}}\subseteq\Delta$. We first prove $\Vform{a}{\Gamma} = \Vform{a}{\Delta}$. For the left-to-right inclusion, let $p\in\Vform{a}{\Gamma}$ be arbitrary, so $p\in\Gamma$ and $p\in\Var{a}$. By the first conjunct of the locality axiom $\Loc$, this yields $\B{a}p\in\Gamma$. But then, by assumption, we also have $p\in\Delta$ and thus $p\in\Vform{a}{\Delta}$. For the right-to-left inclusion, let $p\in\Vform{a}{\Delta}$ be arbitrary, so $p\in\Delta$ and $p\in\Var{a}$. Suppose for a contradiction that $p\notin\Gamma$. Then, clearly, we have $\neg p\in\Gamma$, so it follows $\B{a}\neg p\in\Gamma$ by the second conjunct of $\Loc$. By assumption, we now obtain $\neg p\in\Delta$, which is a contradiction to the consistency of $\Delta$. Therefore, we have $p\in\Gamma$. 

Next, we prove that $\Bform{a}{\Gamma} = \Bform{a}{\Delta}$. For the left-to-right inclusion, let $\varphi$ be an arbitrary formula and assume that $\B{a}\varphi\in\Gamma$. By axiom $\FourB$, this implies $\B{a}\B{a}\varphi\in\Gamma$ and therefore $\B{a}\varphi\in\modform{\Gamma}{\B{a}}$. By assumption, we thus have $\B{a}\varphi\in\Delta$. For the converse inclusion, let $\B{a}\varphi\in\Delta$ be arbitrary and suppose for a contradiction that $\B{a}\varphi\notin\Gamma$. Then, clearly, we have $\neg\B{a}\varphi\in\Gamma$, so it follows $\B{a}\neg\B{a}\varphi\in\Gamma$ by axiom $\FiveB$. By assumption, we now obtain $\neg\B{a}\varphi\in\Delta$, which is a contradiction to the consistency of $\Delta$. Hence, we have $\B{a}\varphi\in\Gamma$.

From $\Vform{a}{\Gamma} = \Vform{a}{\Delta}$ and $\Bform{a}{\Gamma} = \Bform{a}{\Delta}$, it now follows $\locstate{\Gamma}{a} = \locstate{\Delta}{a}$. The proof of part (\ref{lem:equality:local:states:ii}) is similar. However, for this part, one should also use the fact that maximally con\-sis\-tent sets include all instances of axiom scheme $\DB$.   
\end{proof}

\begin{lemma}\label{lem:accessibility:canonical:model}
Let $\Gamma,\Delta\in\MCS$ and let $a\in\Ag$. In the canonical model $\Mc$, we have: 
\begin{enumerate}
\item\label{lem:accessibility:canonical:model:i} $\edge_{\Gamma} \GBrel{a} \edge_{\Delta} \Leftrightarrow \modform{\Gamma}{\B{a}}\subseteq\Delta$,
\item\label{lem:accessibility:canonical:model:ii} $\edge_{\Gamma} \GKrel{a} \edge_{\Delta} \Leftrightarrow \modform{\Gamma}{\K{a}}\subseteq\Delta$.
\end{enumerate}
\end{lemma}

\begin{proof}
We first prove part~(\ref{lem:accessibility:canonical:model:i}). For the left-to-right direction, assume that $\edge_{\Gamma} \GBrel{a} \edge_{\Delta}$. One readily sees that this yields $\locstate{\Gamma}{a} = \locstate{\Delta}{a}$ as well as $\locstate{\Delta}{a}\in\undir{\edge_{\Gamma}}$ and $\locstate{\Delta}{a}\in\tail{\edge_{\Delta}}$. From $\locstate{\Delta}{a}\in\tail{\edge_{\Delta}}$, it follows $\modform{\Delta}{\B{a}}\subseteq\Delta$ by definition of $\edge_{\Delta}$. And since $\locstate{\Gamma}{a} = \locstate{\Delta}{a}$, we also have $\Bform{a}{\Gamma} = \Bform{a}{\Delta}$ and thus $\modform{\Gamma}{\B{a}} = \modform{\Delta}{\B{a}}$. Therefore, $\modform{\Gamma}{\B{a}}\subseteq\Delta$. For the converse direction, suppose that $\modform{\Gamma}{\B{a}}\subseteq\Delta$. By Lemma~\ref{lem:equality:local:states}~(\ref{lem:equality:local:states:i}), we then have $\locstate{\Gamma}{a} = \locstate{\Delta}{a}$. We now prove $\modform{\Delta}{\B{a}}\subseteq\Delta$. To this end, let $\varphi\in\modform{\Delta}{\B{a}}$ be arbitrary and suppose for a contradiction that $\B{a}\varphi\notin\Gamma$. By Lemma~\ref{lem:properties:mcs}, this yields $\neg\B{a}\varphi\in\Gamma$, so it follows $\B{a}\neg\B{a}\varphi\in\Gamma$ by axiom $\FiveB$. Hence, by assumption, $\neg\B{a}\varphi\in\Delta$. But since $\varphi\in\modform{\Delta}{\B{a}}$, we also have $\B{a}\varphi\in\Delta$, in contradiction to the consistency of $\Delta$. Therefore, we must have $\B{a}\varphi\in\Gamma$. By assumption, this yields $\varphi\in\Delta$. Because $\varphi$ was an arbitrary formula from $\modform{\Delta}{\B{a}}$, this shows that $\modform{\Delta}{\B{a}}\subseteq\Delta$. As a consequence, the vertex $\locstate{\Gamma}{a} = \locstate{\Delta}{a}$ must satisfy $\COLc(\locstate{\Gamma}{a}) = a$, $\locstate{\Gamma}{a}\in\undir{\edge_{\Gamma}}$ and $\locstate{\Gamma}{a}\in\tail{\edge_{\Delta}}$, so it follows $\edge_{\Gamma} \GBrel{a} \edge_{\Delta}$. 

Next, we prove part~(\ref{lem:accessibility:canonical:model:ii}). For the left-to-right direction, assume that $\edge_{\Gamma} \GKrel{a} \edge_{\Delta}$. Then, clearly, $\locstate{\Gamma}{a} = \locstate{\Delta}{a}$. In order to prove $\modform{\Gamma}{\K{a}}\subseteq\Delta$, let $\varphi\in\modform{\Gamma}{\K{a}}$ be arbitrary, so $\K{a}\varphi\in\Gamma$. By axioms $\FourK$ and $\KIB$, this yields $\B{a}\K{a}\varphi\in\Gamma$. And since $\locstate{\Gamma}{a} = \locstate{\Delta}{a}$, it also holds $\Bform{a}{\Gamma} = \Bform{a}{\Delta}$, so $\B{a}\K{a}\varphi\in\Delta$. Towards a contradiction, suppose now that $\varphi\notin\Delta$. By axiom $\TK$, we then have $\K{a}\varphi\notin\Delta$, so it follows $\neg\K{a}\varphi\in\Delta$ by Lemma~\ref{lem:properties:mcs}. By axioms $\FiveK$ and $\KIB$, this implies $\B{a}\neg\K{a}\varphi\in\Delta$. From $\B{a}\K{a}\varphi\in\Delta$ and $\B{a}\neg\K{a}\varphi\in\Delta$, we can now deduce $\B{a}(\K{a}\varphi\wedge\neg\K{a}\varphi)\in\Delta$. But by axiom $\DB$, we also have $\neg\B{a}(\K{a}\varphi\wedge\neg\K{a}\varphi)\in\Delta$, which is a contradiction to the consistency of $\Delta$. Therefore, we have $\varphi\in\Delta$. Because $\varphi\in\modform{\Gamma}{\K{a}}$ was arbitrary, this shows that $\modform{\Gamma}{\K{a}}\subseteq\Delta$. For the converse direction, suppose that $\modform{\Gamma}{\K{a}}\subseteq\Delta$. By Lemma~\ref{lem:equality:local:states}~(\ref{lem:equality:local:states:ii}), we then have $\locstate{\Gamma}{a} = \locstate{\Delta}{a}$, so it follows $\edge_{\Gamma} \GKrel{a} \edge_{\Delta}$. 
\end{proof}

We can now prove that $\Mc$ is in fact simple, $n$-uniform and tail-complete. Moreover, a formula $\varphi$ is satisfied by an edge $\edge_{\Gamma}$ in $\Mc$ if and only if $\varphi\in\Gamma$. 

\begin{lemma}\label{lem:properties:canonical:model}
The hypergraph $\Gc$ is simple, $n$-uniform and tail-complete.     
\end{lemma}

\begin{proof}
\emph{Uniformity.} Since we have $\locstate{\Gamma}{a}\neq\locstate{\Gamma}{b}$ for all $\Gamma\in\MCS$ and all $a,b\in\Ag$ with $a\neq b$, every edge $\edge_{\Gamma}$ in $\Gc$ contains exactly $n$ vertices, so $\Gc$ is $n$-uniform. 

\emph{Simplicity.} Towards a contradiction, suppose that there are $\Gamma,\Delta\in\MCS$ with $\edge_{\Gamma}\neq\edge_{\Delta}$ such that $\undir{\edge_{\Gamma}}\subseteq\undir{\edge_{\Delta}}$. By the $n$-uniformity of $\Gc$, this is only possible if $\undir{\edge_{\Gamma}} = \undir{\edge_{\Delta}}$, so we may conclude that it holds $\locstate{\Gamma}{a} = \locstate{\Delta}{a}$ for all $a\in\Ag$. By Lemma~\ref{lem:equality:of:mcs}, we now obtain $\Gamma=\Delta$. But this implies $\edge_{\Gamma} = \edge_{\Delta}$, which is a contradiction to our assumption. Therefore, $\Gc$ must be simple. 

\emph{Tail-completeness.} Let $\locstate{\Gamma}{a}\in\Vertices(\Gc)$ be arbitrary. We prove that there exists a hyperedge $\edge_{\Delta}$ with $\locstate{\Gamma}{a}\in\tail{\edge_{\Delta}}$. Suppose for a contradiction that $\modform{\Gamma}{\B{a}}$ is not consistent, so there is a finite subset $\Pi\subseteq\modform{\Gamma}{\B{a}}$ with $\EDL\prov\bigwedge\Pi\rightarrow\bot$. Let $\varphi := \bigwedge\Pi$. Then $\EDL\prov\neg\varphi$, so we obtain $\EDL\prov\B{a}\neg\varphi$ by necessitation and axiom $\KIB$. Hence, by Lemma~\ref{lem:properties:mcs}~(\ref{lem:properties:mcs:ii}), we have $\B{a}\neg\varphi\in\Gamma$. Since $\Pi\subseteq\modform{\Gamma}{\B{a}}$, we also have $\B{a}\psi\in\Gamma$ for all $\psi\in\Pi$, so it follows $\B{a}\varphi\in\Gamma$. Together with $\B{a}\neg\varphi\in\Gamma$, this yields $\B{a}(\varphi\wedge\neg\varphi)\in\Gamma$. But by axiom $\DB$, we have $\neg\B{a}(\varphi\wedge\neg\varphi)\in\Gamma$, which is a contradiction to the consistency of $\Gamma$. Hence, $\modform{\Gamma}{\B{a}}$ is consistent, so it can be extended to a maximally consistent set $\Delta$ by Lemma~\ref{lem:lindenbaum}. Since $\modform{\Gamma}{\B{a}}\subseteq\Delta$, we have $\edge_{\Gamma}\GBrel{a}\edge_{\Delta}$ by Lemma~\ref{lem:accessibility:canonical:model}~(\ref{lem:accessibility:canonical:model:i}). But then, clearly, $\locstate{\Gamma}{a}$ satisfies $\locstate{\Gamma}{a}\in\tail{\edge_{\Delta}}$. Because $\locstate{\Gamma}{a}\in\Vertices(\Gc)$ was arbitrary, this shows that $\Gc$ is tail-complete. 
\end{proof}

\begin{lemma}[Truth Lemma]\label{lem:truth:lemma}
For every maximally consistent set $\Gamma\in\MCS$ and for every formula $\varphi\in\LKB$, we have: $\Mc,\edge_{\Gamma}\semH\varphi$ iff $\varphi\in\Gamma$.
\end{lemma}

\begin{proof}
By induction on the structure of $\varphi$. We only consider the case in which $\varphi$ is of the form $\varphi = \B{a}\psi$ for some $a\in\Ag$. For the left-to-right direction, assume that $\Mc,\edge_{\Gamma}\semH\B{a}\psi$. Suppose for a contradiction that the set $(\modform{\Gamma}{\B{a}})\cup\{\neg\psi\}$ is consistent. Then, by Lemma~\ref{lem:lindenbaum}, it can be extended to a maximally consistent set $\Delta$. Since $\neg\psi\in\Delta$, we have $\psi\notin\Delta$, so it follows $\Mc,\edge_{\Delta}\not\semH\psi$ by the induction hypothesis. And from $\modform{\Gamma}{\B{a}}\subseteq\Delta$, we obtain $\edge_{\Gamma} \GBrel{a} \edge_{\Delta}$ by Lemma~\ref{lem:accessibility:canonical:model}~(\ref{lem:accessibility:canonical:model:i}). But then $\Mc,\edge_{\Gamma}\not\semH\B{a}\psi$, in contradiction to our assumption. Therefore, $(\modform{\Gamma}{\B{a}})\cup\{\neg\psi\}$ is not consistent, so there is a finite subset $\Sigma\subseteq (\modform{\Gamma}{\B{a}})\cup\{\neg\psi\}$ such that $\EDL\prov \bigwedge\Sigma\rightarrow\bot$. Without loss of generality, we may assume that $\Sigma$ is of the form $\Sigma = \Pi\cup\{\neg\psi\}$, where $\Pi\subseteq\modform{\Gamma}{\B{a}}$. By propositional reasoning, we now obtain $\EDL\prov\pi\rightarrow\psi$, where $\pi := \bigwedge\Pi$. This yields $\EDL\prov\B{a}\pi\rightarrow\B{a}\psi$ by necessitation and axioms $\KIB$ and $\KB$, so we may conclude $(\B{a}\pi\rightarrow\B{a}\psi)\in\Gamma$ by Lemma~\ref{lem:properties:mcs}~(\ref{lem:properties:mcs:ii}). Since $\Pi\subseteq\modform{\Gamma}{\B{a}}$, we also have $\B{a}\sigma\in\Gamma$ for all $\sigma\in\Pi$. This implies $\B{a}\pi\in\Gamma$. Using Lemma~\ref{lem:properties:mcs}~(\ref{lem:properties:mcs:iii}), we now obtain $\B{a}\psi\in\Gamma$, as desired. 

For the converse direction, assume $\B{a}\psi\in\Gamma$. Let $\Delta\in\MCS$ be arbitrary such that $\edge_{\Gamma} \GBrel{a} \edge_{\Delta}$. By Lemma~\ref{lem:accessibility:canonical:model}~(\ref{lem:accessibility:canonical:model:i}), we then have $\modform{\Gamma}{\B{a}}\subseteq\Delta$. Thus, from $\B{a}\psi\in\Gamma$, we obtain $\psi\in\Delta$. This yields $\Mc,\edge_{\Delta}\semH\psi$ by the induction hypothesis. Since $\Delta\in\MCS$ was arbitrary with $\edge_{\Gamma} \GBrel{a} \edge_{\Delta}$, this shows that $\Mc,\edge_{\Gamma}\semH\B{a}\psi$.     
\end{proof}

The desired completeness result now follows as an immediate consequence. 

\begin{theorem}[Completeness of $\EDL$]
The proof system $\EDL$ is sound and complete with respect to the class $\HSUTclass$ of all simple, $n$-uniform and tail-complete hypergraph models. That is, for every $\varphi\in\LKB$, we have: $\EDL\prov\varphi$ iff $\HSUTclass\semH\varphi$. 
\end{theorem}

\begin{proof}
The soundness can be established by a routine induction on the length of a proof in $\EDL$. For the completeness direction, suppose that $\EDL\notprov\varphi$. Then, clearly, the set $\{\neg\varphi\}$ must be consistent, so it can be extended to a maximally consistent set $\Gamma$ by Lemma~\ref{lem:lindenbaum}. Since $\neg\varphi\in\Gamma$, we must have $\Mc,\edge_{\Gamma}\semH\neg\varphi$ by Lemma~\ref{lem:truth:lemma}, so we may conclude that $\Mc,\edge_{\Gamma}\not\semH\varphi$. Because $\Mc$ is simple, $n$-uni\-form and tail-complete by Lemma~\ref{lem:properties:canonical:model}, this shows that $\HSUTclass\not\semH\varphi$.  
\end{proof}

As mentioned above, an almost identical canonical model construction can also be used to show that $\LocKFourFive$ is complete with respect to the class $\HSUclass$ and that $\LocKDFourFive$ is complete with respect to the class $\HSUTclass$. The only difference is that maximally consistent sets are now restricted to formulas from the fragment $\LB$ and defined with reference to the systems $\LocKFourFive$ and $\LocKDFourFive$, respectively. Using the same reasoning as in Lemma~\ref{lem:properties:canonical:model}, one can prove that the resulting canonical models indeed have the desired properties. That is, the canonical model for $\LocKFourFive$ is simple and $n$-uniform, and the one for $\LocKDFourFive$ is simple, $n$-uniform and tail-complete. The truth lemma can be established in the same way as in Lemma~\ref{lem:truth:lemma}.   

\begin{theorem}[Further Completeness Results]
The system $\LocKFourFive$ is sound and complete with respect to the class $\HSUclass$ and the system $\LocKDFourFive$ is sound and complete with respect to the class $\HSUTclass$. That is, for every formula $\varphi$ from the doxastic fragment $\LB$, we have the equivalence $\LocKFourFive\prov\varphi \Leftrightarrow \HSUclass\semH\varphi$ and the equivalence $\LocKDFourFive\prov\varphi \Leftrightarrow \HSUTclass\semH\varphi$. 
\end{theorem}

\section{Correspondence to Kripke Models}\label{sec:conversions:kripke:hypergraph}

We will now present direct conversions between doxastic Kripke models and directed hypergraph models. More specifically, we will see that every Kripke model from the class $\KTEclass$ can be transformed into an equivalent hypergraph model from the class $\HSUclass$ and vice versa. The same correspondence also holds between the smaller classes $\KSTEclass$ and $\HSUTclass$. An illustration of the correspondence between Kripke models and hypergraph models is provided in Appendix~\ref{app:illustration:conversion}.\medskip 

\noindent\emph{From Kripke models to hypergraph models.} We first introduce a canonical conversion from doxastic Kripke models to hypergraph models. The procedure works in essentially the same way as the usual transformation from local and proper epistemic Kripke models to pure simplicial models (see \cite[Sects.~2.1--2.2]{goubault:etal:2021}).

\begin{definition}
Let $M = (W,\DoxRels,V)$ be a local doxastic Kripke model. For every pair $(a,u)\in\Ag\times W$, let $x_u^a$ be a fresh vertex and put $\Vertices_M := \{x_u^a \mid \text{$a\in\Ag$, $u\in W$}\}$. Moreover, let $\auxrel\subseteq \Vertices_M\times \Vertices_M$ be the equivalence relation given by $\auxrel := \{(x_u^a, x_v^a) \mid \text{$a\in\Ag$, $(u,v)\in\eqrel{\Brel{a}}$}\}$ and let $[x_u^a]$ be the equivalence class of a vertex $x_u^a$ with respect to $\auxrel$.\footnote{Observe that each of the classes $[x_u^a]$ must be non-empty by the reflexivity of $\eqrel{\Brel{a}}$.} For any $u\in W$, we also define $X_u := \{[x_u^a] \mid \text{$a\in\Ag$, $(u,u)\in \Brel{a}$}\}$ and $Y_u := \{[x_u^a] \mid \text{$a\in\Ag$, $(u,u)\notin \Brel{a}$}\}$. The \emph{hypergraph model associated with $M$} is the triple $\hypmod{M} =(G,\colmap,\ell)$ defined as follows:
\begin{itemize}
\item $G = (\Vertices,\Edges)$, where $\Vertices = \{[x_u^a] \mid x_u^a\in \Vertices_M\}$ and $\Edges = \{(X_u, Y_u)\mid u\in W\}$,
\item $\colmap([x_u^a]) := a$, for all $a\in\Ag$ and all $u\in W$,
\item $\ell([x_u^a]) := V(u)\cap\Var{a}$, for all $a\in\Ag$ and all $u\in W$.
\end{itemize}
For any $u\in W$, we also write $\edge_u$ for the hyperedge given by $\edge_u :=(X_u,Y_u)$. 
\end{definition}

We first note that $\hypmod{M} =(G,\colmap,\ell)$ is indeed a well-defined hypergraph model, for every local doxastic Kripke model $M = (W,\DoxRels,V)$. First of all, by construction, every hyperedge $\edge_u =(X_u,Y_u)$ in $\hypmod{M}$ must contain at most one equivalence class $[x_u^a]$ for each $a\in\Ag$, so the function $\colmap$ cannot assign the same color to multiple vertices within the same hyperedge. In other words, $\chi$ is in fact a \emph{coloring} of $G$. Now, to see that the valuation $\ell$ is well-defined, let $x_u^a, x_v^a \in \Vertices_M$ be arbitrary such that $[x_u^a] = [x_v^a]$. By definition of $\auxrel$, this yields $(u,v)\in\eqrel{\Brel{a}}$ and thus $(u,v)\in (\sym{\Brel{a}})^{m}$ for some $m\in\mathbb{N}$. Using the locality of $M$ and an easy induction on $m$, one can now prove that $V(u)\cap\Var{a} = V(v)\cap\Var{a}$. Therefore, $\ell$ is a well-defined function from equivalence classes to sets of local variables.

\begin{proposition}\label{prop:properties:kripke:to:hypergraph}
Let $M$ be a local doxastic Kripke model. 
\begin{enumerate}
\item\label{prop:properties:kripke:to:hypergraph:i} If $M$ belongs to the class $\KTEclass$, then $\hypmod{M}$ belongs to the class $\HSUclass$.
\item\label{prop:properties:kripke:to:hypergraph:ii} If $M$ belongs to the class $\KSTEclass$, then $\hypmod{M}$ belongs to the class $\HSUTclass$.
\end{enumerate}
\end{proposition}

A proof is provided in Appendix~\ref{app:prop:kripke:to:hyper}. One can also show that $M$ and $\hypmod{M}$ are modally equivalent in the sense that a formula $\varphi$ is true at a world $w$ in $M$ if and only if $\varphi$ is satisfied by the corresponding hyperedge $\edge_w$ in $\hypmod{M}$. If $M$ is not assumed to be serial, then the equivalence only holds for formulas from the fragment $\LB$. Otherwise, we consider the full language $\LKB$. The proof proceeds by an easy induction on the structure of a formula and is therefore omitted.  

\begin{proposition}
Let $M$ be a doxastic Kripke model and let $w$ be a world in $M$.
\begin{itemize}
\item If $M\in\KTEclass$, then for all $\varphi\in\LB$, we have: $M,w\semK\varphi \Leftrightarrow \hypmod{M},\edge_w\semH\varphi$.
\item If $M\in\KSTEclass$, then for all $\varphi\in\LKB$, we have: $M,w\semK\varphi \Leftrightarrow \hypmod{M},\edge_w\semH\varphi$.%
\end{itemize}
\end{proposition}

\noindent\emph{From hypergraph models to Kripke models.} For the converse direction, we now define a general transformation from hypergraph models to Kripke models. We then show that models from the class $\HSUclass$ are indeed transformed into models from the class $\KTEclass$, while those from $\HSUTclass$ are transformed into models from $\KSTEclass$.

\begin{definition}
Let $M=(G,\colmap,\ell)$ be a hypergraph model. The \emph{Kripke model associated with $M$} is the triple $\kripmod{M} = (W,\DoxRels,V)$ defined in the following way: 
\begin{itemize}
\item $W:=\Edges(G)$ is the set of hyperedges of $G$,
\item $(\edge_1,\edge_2)\in\Brel{a} :\Leftrightarrow \edge_1\GBrel{a}\edge_2$, for all $\edge_1,\edge_2\in W$,
\item $V(\edge) := \ell(\edge)$, for all $\edge\in W$. 
\end{itemize}
\end{definition} 

\begin{proposition}\label{prop:properties:hypergraph:to:kripke}
Let $M$ be a hypergraph model. 
\begin{enumerate}
\item\label{prop:properties:hypergraph:to:kripke:i} If $M$ belongs to the class $\HSUclass$, then $\kripmod{M}$ belongs to the class $\KTEclass$. 
\item\label{prop:properties:hypergraph:to:kripke:ii} If $M$ belongs to the class $\HSUTclass$, then $\kripmod{M}$ belongs to the class $\KSTEclass$.
\end{enumerate} 
\end{proposition}

A proof can be found in Appendix~\ref{app:prop:hyper:to:kripke}. It remains to show that $M$ and $\kripmod{M}$ are modally equivalent: a formula $\varphi$ is satisfied at a hyperedge $\edge$ in $M$ if and only if $\varphi$ is satisfied at $\edge$ in $\kripmod{M}$. The proof is again easy and therefore omitted.   

\begin{proposition}
Let $M$ be a hypergraph model and let $\edge$ be a hyperedge in $M$.
\begin{itemize}
\item If $M\in\HSUclass$, then for all $\varphi\in\LB$, we have: $M,\edge\semH\varphi \Leftrightarrow \kripmod{M},\edge\semK\varphi$.
\item If $M\in\HSUTclass$, then for all $\varphi\in\LKB$, we have: $M,\edge\semH\varphi \Leftrightarrow\kripmod{M},\edge\semK\varphi$.
\end{itemize}
\end{proposition}

\section{Conclusion}\label{sec:conclusion}

We presented several types of directed hypergraphs as models for a system of epis\-temic-dox\-as\-tic logic and for two pure doxastic logics. We also proved that these logics are sound and complete with respect to certain classes of directed hypergraph models. Our work proposes an intuitive solution to the problem of modeling false beliefs raised in \cite[Sect.~4.3]{castaneda:etal:2024}. For future work, we plan to consider different accessibility relations on directed hyperedges. First investigations suggest that this also yields a hypergraph semantics for logics violating negative introspection. Furthermore, we plan to explore belief dynamics on hypergraph models, thereby addressing the next open problem from \cite{castaneda:etal:2024}: modeling agents capable of lying and deceiving others. A concrete example would be a scenario in which malicious agents may lie about their preferences in a voting scheme to influence the decision of others. In such a setting, honest agents must act on the basis of their beliefs, since they cannot know whether they have been lied to.

\paragraph{Acknowledgments.} Djanira Gomes, Valentin M{\"u}ller and Thomas Studer were supported by the Swiss National Science Foundation (SNSF) under grant number 10000440 (Epistemic Group Attitudes). David Lehnherr was supported by the SNSF under grant number 219403 (Emerging Consensus).

\newpage 
\appendix
\section{Supplement to Section~\ref{sec:conversions:kripke:hypergraph}}\label{app:illustration:conversion}

The illustration below shows a proper doxastic Kripke frame (left) and the corresponding chromatic hypergraph (right). As can be seen, the frame is serial, transitive and Euclidean, and the hypergraph is simple, $3$-uniform and tail-complete. 
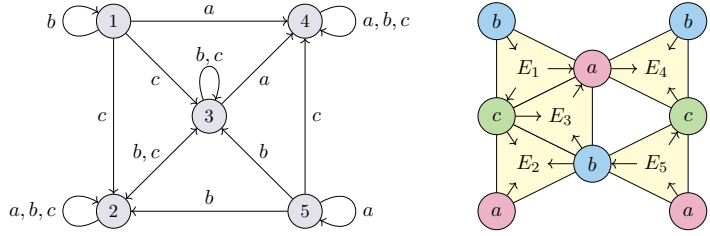
\begin{center}
\scalebox{.7}{
\begin{tikzpicture}[scale=0.9]
\tikzset{every loop/.style={min distance=10mm,looseness=8}}

\node[circle,draw,font=\small,fill=mynodecolor] (w1) at (0,4) {$1$};
\node[circle,draw,font=\small,fill=mynodecolor] (w2) at (0,0) {$2$};
\node[circle,draw,font=\small,fill=mynodecolor] (w3) at (2,2) {$3$};
\node[circle,draw,font=\small,fill=mynodecolor] (w4) at (4,4) {$4$};
\node[circle,draw,font=\small,fill=mynodecolor] (w5) at (4,0) {$5$};

\draw[->] (w1)--(w2) node[midway,left] {$c$};
\draw[->] (w1)--(w3) node[midway,below left,xshift=1mm] {$c$};
\draw[->] (w1)--(w4) node[midway,above] {$a$};
\draw[->] (w1) edge[out=210,in=150,loop] node[left] {$b$} ();
\draw[->] (w2) edge[out=210,in=150,loop] node[left] {$a,b,c$} ();
\draw[<->] (w2)--(w3) node[midway,above left,xshift=1mm,yshift=-1mm] {$b,c$};
\draw[->] (w3)--(w4) node[midway,below right,xshift=-1mm] {$a$};
\draw[->] (w3) edge[out=110,in=70,loop] node[above,yshift=-1mm] {$b,c$} ();
\draw[->] (w4) edge[out=30,in=-30,loop] node[right] {$a,b,c$} ();
\draw[->] (w5)--(w2) node[midway,above] {$b$};
\draw[->] (w5)--(w3) node[midway,above right,xshift=-1mm,yshift=-1mm] {$b$\vphantom{$,$}};
\draw[->] (w5)--(w4) node[midway,right] {$c$};
\draw[->] (w5) edge[out=30,in=-30,loop] node[right] {$a$} ();

\tikzset{minimum size=.7cm}

\node[circle,draw,fill=myred] (a1) at (8,0) {$a$};
\node[circle,draw,fill=mygreen] (c1) at (8,2) {$c$};
\node[circle,draw,fill=myblue] (b1) at (8,4) {$b$};
\node[circle,draw,fill=myred] (a2) at (10,3) {$a$};
\node[circle,draw,fill=myblue] (b2) at (10,1) {$b$};
\node[circle,draw,fill=myred] (a3) at (12,0) {$a$};
\node[circle,draw,fill=mygreen] (c2) at (12,2) {$c$};
\node[circle,draw,fill=myblue] (b3) at (12,4) {$b$};

\node (E1) at (barycentric cs:a2=1,b1=1,c1=1) {$\edge_1$};
\node (E2) at (barycentric cs:a1=1,b2=1,c1=1) {$\edge_2$};
\node (E3) at (barycentric cs:a2=1,b2=1,c1=1) {$\edge_3$};
\node (E4) at (barycentric cs:a2=1,b3=1,c2=1) {$\edge_4$};
\node (E5) at (barycentric cs:a3=1,b2=1,c2=1) {$\edge_5$};

\begin{scope}[on background layer]
\filldraw[fill=mybgcolor] (a2.center)--(b1.center)--(c1.center)--cycle;
\filldraw[fill=mybgcolor] (a1.center)--(b2.center)--(c1.center)--cycle;
\filldraw[fill=mybgcolor] (a2.center)--(b2.center)--(c1.center)--cycle;
\filldraw[fill=mybgcolor] (a2.center)--(b3.center)--(c2.center)--cycle;
\filldraw[fill=mybgcolor] (a3.center)--(b2.center)--(c2.center)--cycle;
\end{scope} 

\draw[->] (E1) edge (a2) (b1) edge (E1) (E1) edge (c1);
\draw[->] (a1) edge (E2) (b2) edge (E2) (c1) edge (E2);
\draw[->] (E3) edge (a2) (b2) edge (E3) (c1) edge (E3);
\draw[->] (a2) edge (E4) (b3) edge (E4) (c2) edge (E4);
\draw[->] (a3) edge (E5) (E5) edge (b2) (E5) edge (c2);

\node (p) at (13.5,0) {};
\end{tikzpicture}}
\end{center}

\section{Proof of Proposition~\ref{prop:properties:kripke:to:hypergraph}}\label{app:prop:kripke:to:hyper}

We only prove part~(\ref{prop:properties:kripke:to:hypergraph:ii}) of the proposition. Let $M = (W,\DoxRels,V)$ be a doxastic Kripke model with $M\in\KSTEclass$, so $M$ is local and proper, and each of the relations $\Brel{a}$ is serial, transitive and Euclidean. Let $\hypmod{M} =(G,\colmap,\ell)$ be the associated hypergraph model. We prove that $\hypmod{M}$ is simple, $n$-uniform and tail-complete.

\emph{Uniformity.} Since $[x_u^a]\neq [x_u^b]$ for all $a,b\in\Ag$ with $a\neq b$, every $\edge_u$ contains exactly one class $[x_u^a]$ for each $a\in\Ag$. Thus, $|\undir{\edge_u}| = n$ for all $\edge_u\in\Edges(G)$.

\emph{Simplicity.} Let $\edge_u,\edge_v\in\Edges(G)$ be arbitrary hyperedges with $\edge_u \neq\edge_v$ and consider the corresponding worlds $u,v\in W$. Since $\edge_u \neq\edge_v$, we must have $u\neq v$. Hence, because $M$ is proper, there exists some agent $a\in \Ag$ with $(u,v)\notin \eqrel{\Brel{a}}$ and thus $[x_u^a]\neq [x_v^a]$. Suppose for a contradiction that $\undir{\edge_u}\subseteq\undir{\edge_v}$. Since $\hypmod{M}$ is $n$-uniform, this implies $\undir{\edge_u}=\undir{\edge_v}$. Hence, we must have $[x_u^a] = [x_v^a]$, in contradiction to what was said above. Therefore, we have $\undir{\edge_u}\not\subseteq\undir{\edge_v}$.
  
\emph{Tail-completeness.} Let $[x_u^a]\in\Vertices(G)$ be arbitrary. By the seriality of $\Brel{a}$, there exists some world $w\in W$ with $(u,w)\in\Brel{a}$. Since $\Brel{a}$ is Euclidean, this yields $(w,w)\in\Brel{a}$. Hence, we have $[x_w^a]\in X_w$, so $[x_w^{a}]\in\tail{\edge_w}$. Because $(u,w)\in\Brel{a}$, we also have $(u,w)\in\eqrel{\Brel{a}}$. This implies $[x_u^a] = [x_w^a]$ and therefore $[x_u^a]\in\tail{\edge_w}$. Since $[x_u^a]\in\Vertices(G)$ was arbitrary, this shows that $\Vertices(G) = \bigcup_{\edge\in\Edges(G)}\tail{\edge}$.\qed

\section{Proof of Proposition~\ref{prop:properties:hypergraph:to:kripke}}\label{app:prop:hyper:to:kripke}

We only prove part (\ref{prop:properties:hypergraph:to:kripke:ii}) of the proposition. Let $M=(G,\colmap,\ell)$ be a hypergraph model and suppose that $M\in \HSUTclass$, so $M$ is simple, $n$-uni\-form, and tail-complete. Moreover, let $\kripmod{M} = (W,\DoxRels,V)$ be the Kripke model determined by $M$. By Proposition~\ref{prop:properties:defined:relations}, we already know that each of the relations $\Brel{a}$ is serial, transitive and Euclidean. Hence, it suffices to establish the following properties.

\emph{Locality.} Let $a\in\Ag$ and $\edge_1,\edge_2\in W$ be arbitrary such that $(\edge_1,\edge_2)\in\sym{\Brel{a}}$. This implies $\edge_1\GBrel{a}\edge_2$ or $\edge_2\GBrel{a}\edge_1$. In either case, $\edge_1$ and $\edge_2$ must share an $a$-colored vertex $u$, so it follows that $V(\edge_1)\cap\Var{a} = \ell(u) = V(\edge_2)\cap\Var{a}$.

\emph{Properness.} Suppose for a contradiction that there are $\edge_1,\edge_2\in W$ with $\edge_1\neq\edge_2$ such that $(\edge_1,\edge_2)\in\eqrel{\Brel{a}}$ for all $a\in\Ag$. It is easy to see that, in this case, $\edge_1$ and $\edge_2$ must contain exactly the same vertices, so $\undir{\edge_1} = \undir{\edge_2}$. But this contradicts the assumption that $M$ is simple. Hence, $\kripmod{M}$ is proper.\qed

\bibliographystyle{plainurl}
\bibliography{bibliography}

\end{document}